\documentclass{amsart}
\usepackage{amssymb,amsfonts}
\usepackage[all,arc]{xy}
\usepackage{enumerate}
\usepackage{color,soul}
    \sethlcolor{white}
\usepackage{mathrsfs}
\usepackage{pinlabel}
\usepackage{mathptmx}
\usepackage{anyfontsize}
\usepackage{t1enc}
\usepackage{rotating}
\usepackage{xparse}
\usepackage{subcaption,graphicx}
\newtheorem{thm}{Theorem}[section]
\newtheorem{res}{Result}[section]

\newtheorem{prop}[thm]{Proposition}
\newtheorem{lem}[thm]{Lemma}

\usepackage{graphicx}
\theoremstyle{definition}

\newtheorem{exmp}[thm]{Example}

\theoremstyle{remark}

\makeatletter
\let\c@equation\c@thm
\makeatother
\numberwithin{equation}{section}

\title{Modeling knotted proteins with tangles}
\author{
  Isabel K. Darcy, Garret Jones, and
 Puttipong Pongtanapaisan}

\begin{document}

 \begin{abstract}
Although rare, an increasing number of proteins have been observed to contain entanglements in their native structures. To gain more insight into the significance of protein knotting, researchers have been investigating protein knot formation using both experimental and theoretical methods. Motivated by the hypothesized folding pathway of $\alpha$-haloacid dehalogenase (DehI) protein, Flapan, He, and Wong proposed a theory of how protein knots form, which includes existing folding pathways described by Taylor and B{\"o}linger et al. as special cases. In their topological descriptions, two loops in an unknotted open protein chain containing at most two twists each come close together, and one end of the protein eventually passes through the two loops. In this paper, we build on Flapan, He, and Wong's theory where we pay attention to the crossing signs of the threading process and assume that the unknotted protein chain may arrange itself into a more complicated configuration before threading occurs. We then apply tangle calculus, originally developed by Ernst and Sumners to analyze the action of specific proteins on DNA, to give all possible knots or knotoids that may be discovered in the future according to our model and give recipes for engineering specific knots in proteins from simpler pieces. We show why twists knots are the most likely knots to occur in proteins.  
We use chirality to show that the most likely knots to occur in proteins via Taylor's twisted hairpin model are the knots $+3_1$, $4_1$, and $-5_2$.
\end{abstract}

\maketitle
\section{Introduction}

Living organisms depend on proteins for a wide range of functions. Among many other important roles, proteins can provide structural stability to cells, catalyze metabolic reactions, recognize pathogens, and transcribe DNA \cite{lesk2010introduction,price2009exploring}. A protein's function is largely dependent on its structure, and proteins' failure to fold properly can cause many diseases including Alzheimer's disease and Parkinson's disease \cite{selkoe2003folding,flapan2015knots}. Therefore, it is of great interest to know exactly how a protein adopts its native three-dimensional state. Theory has suggested that the folding occurs according to a funneled, minimally frustrated energy landscape in which the native state is at the bottom \cite{bryngelson1987spin, mallam2009does}. Despite the many possible conformations, under proper conditions, naturally selected proteins reliably fold into their native conformation. In general, this process is relatively rapid: on the scale of microseconds for some proteins. 
Given the rapid and spontaneous manner in which many proteins fold, one would not expect complicated structures like knots to occur in native states. Nevertheless, research has shown that some proteins do indeed form local knots; some with as many as seven crossings \cite{brems2022alphafold}.

There are several models researchers have used to study protein knots. The first option is to connect the terminal ends of the protein to obtain a circle in such a way that the existing entanglement is preserved and the knot type can be identified \cite{millett2005tying,lua2006statistics,millett2013identifying,alexander2017proteins,goundaroulis2017studies}. A mathematical \textit{knot} is an embedding of a circle into three-dimensional space. 
Fig. \ref{fig:proteinknots} \cite{mansfield1994there,taylor2000deeply,virnau2006intricate,bolinger2010stevedore}(a)-(e) depicts five knot types that can be found in the protein data base (PDB). There are also knots with six and seven crossings whose structures were determined computationally via the machine learning protein prediction method AlphaFold \cite{jumper2021highly,brems2022alphafold}. For example, the knot shown in Fig. \ref{fig:proteinknots}(f) was found in the backbone of the von Willebrand factor A domain-containing protein 5A (BCSC-1, breast cancer suppressor candidate 1) \cite{perlinska2022new}.

To avoid the ambiguity caused by closing up the open ends of a protein, some researchers have also been treating a protein conformation by projecting the 3D shape of a protein onto a sphere or a plane. The result is an open arc, and one declares two projections to be equal if they differ by a finite sequence of combinatorial local moves. These moves are applied away from the end points since any open curve can be untangled if one allows the open end points to move freely. Open arcs up to this equivalence relation were originally studied by Turaev and they are referred to as \textit{knotoids} \cite{turaev2012knotoids}.

 A knot is identified using the notation $n_k$ where $n$ is called the \textit{crossing number}, which is the minimum number of crossings needed to represent the knot as a projection in the 2D-plane. The subscript $k$ means this knot is the $kth$ knot in a list of $n$ crossing knots per the tables at KnotInfo \cite{livingston2020knotinfo}. A knot can have many configurations, and two conformations are \textit{equivalent} if we can deform one into the other without cutting and regluing. When a knot is not equivalent to its mirror image we differentiate them by a plus/minus sign \cite{liang1994left}. The crossing number of knotoids can be defined analogously and a tabulation of knotoids was created by Goundaroulis, Dorier, and Stasiak \cite{goundaroulis2019systematic}.

\begin{figure}
     \centering
     \begin{subfigure}[b]{0.9\textwidth}
         \centering
         \includegraphics[width=\textwidth]{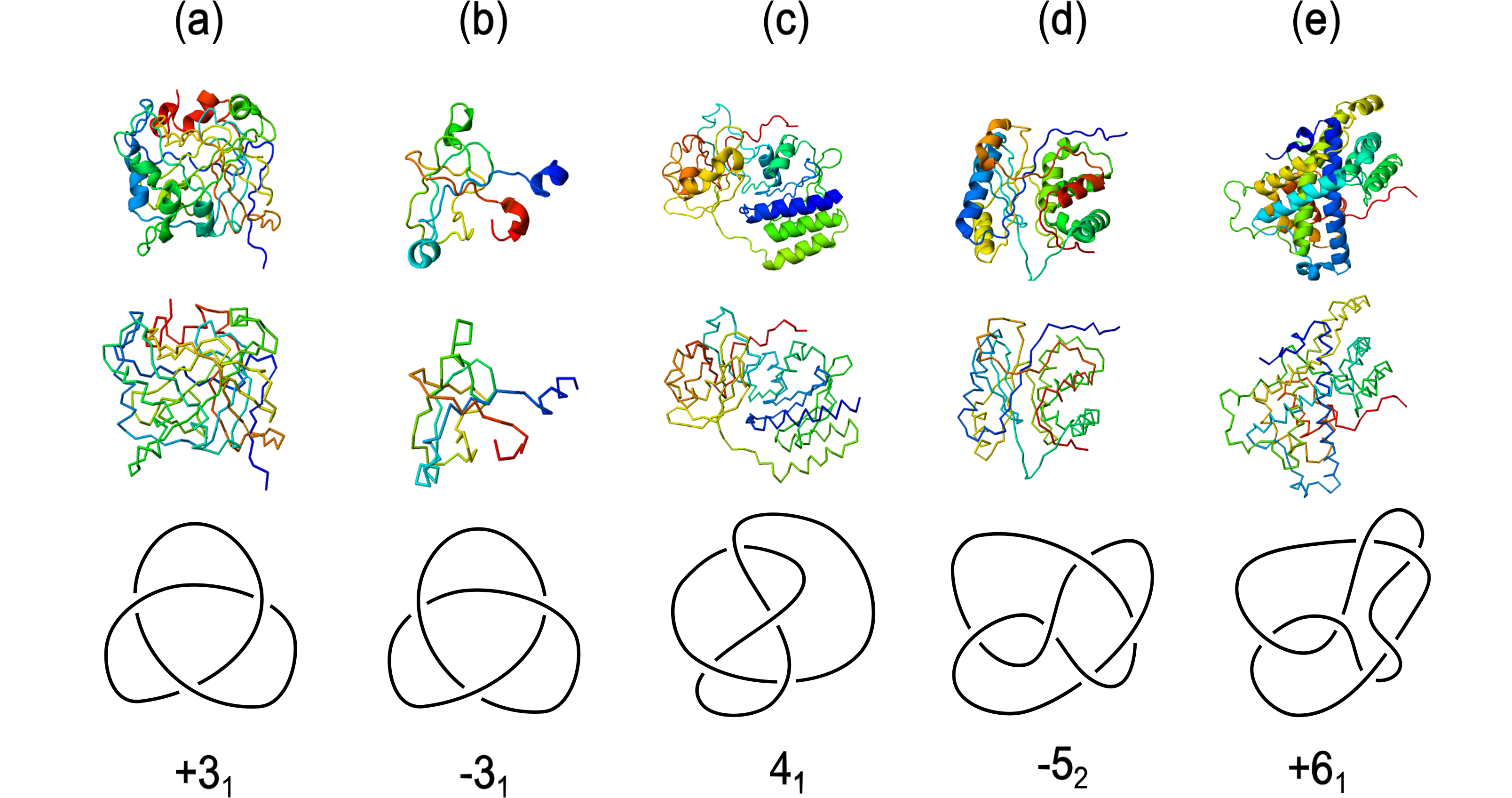}
         \label{fig:5crossings}
     \end{subfigure}
      \begin{subfigure}[b]{0.6\textwidth}
         \centering
         \includegraphics[width=\textwidth]{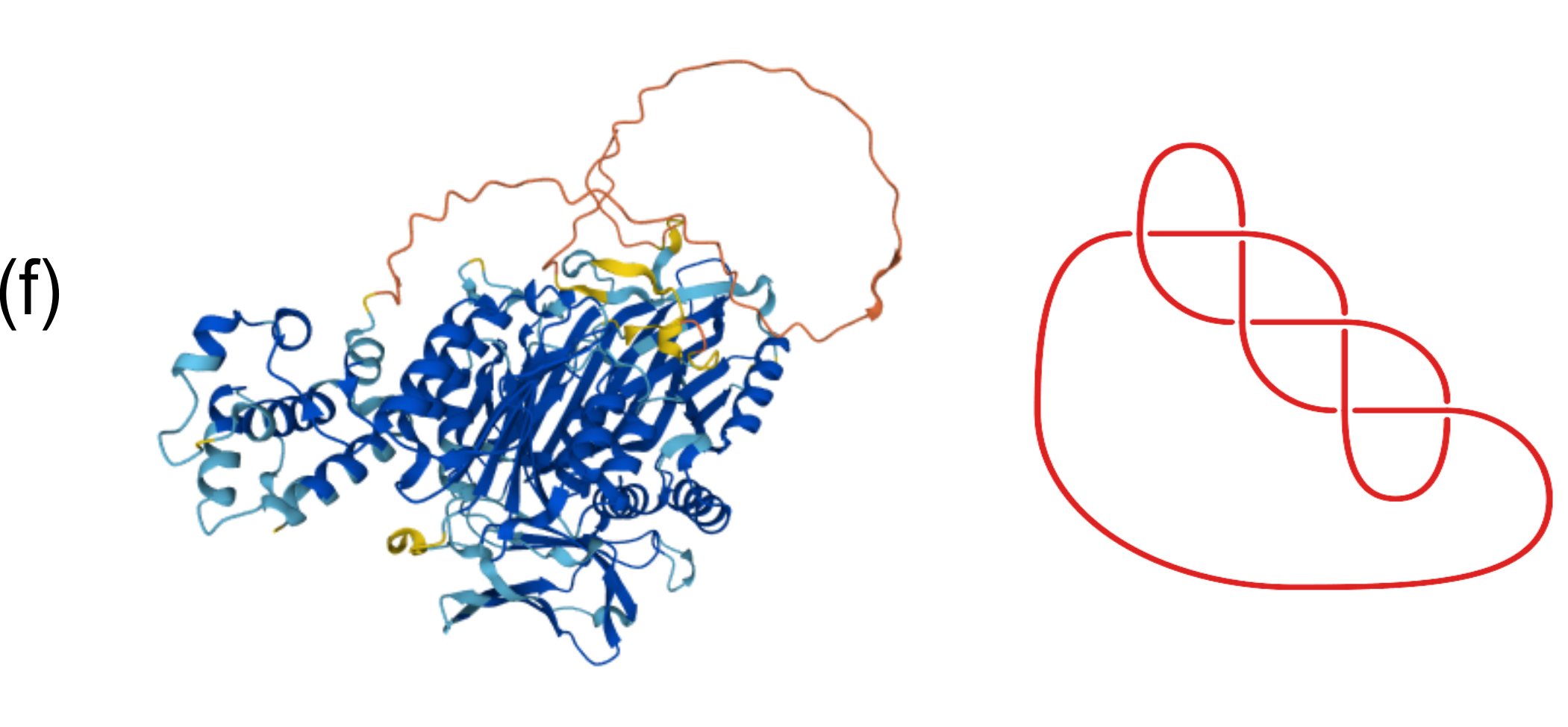}
         \label{fig:6crossings}
     \end{subfigure}

       \caption{Top panel: each ribbon diagram represents an example of a polypeptide chain containing one of the five knot types discovered in proteins. Middle panel: the reduce backbone representations of the examples from the top panel. Images in the top and middle panels were generated using KnotProt \cite{dabrowski2019knotprot}. Bottom panel: the five topological knot types obtained from the middle panel by performing probabilistic closure and deformations without cutting and regluing.
 (a) The alpha carbonic anhydrase from schistosoma mansoni containing a $+3_1$ knot, PDB code: 6QQM. (b) A human fully-assembled precatalytic spliceosome containing a $-3_1$ knot, PDB code: 6QX9. (c) Monomeric near-infrared fluorescent protein containing a $4_1$ knot, PDB code: 5VIK. (d) Solution structure of uchl1 s18y variant containing a $-5_2$ knot, PDB code 2LEN. (e) The $\alpha$-haloacid dehalogenase containing a $+6_1$ knot, PDB code: 4N2X. (f) The $6_3$ knot found in the backbone of the von Willebrand factor A domain containing 5A.}
\label{fig:proteinknots}
\end{figure}

Taylor \cite{taylor2007protein} proposed that a pathway for the formation of a protein knot begins with the protein taking the shape of a twisted bobby pin with a distinguished loop at the top. A terminus then moves towards such a loop, and eventually threads through it. Following Flapan et al., this folding mechanism will be referred to as the \textit{twisted hairpin pathway} \cite{flapan2019topological}. While it is possible that Taylor's theory gives the primary folding pathway for most protein knots, several computational and experimental results hint at the existence of other folding mechanisms \cite{flapan2019topological,bolinger2010stevedore}. 

We incorporate the notion of tangles to model knotted proteins.
An \textit{$n$-string tangle} is a 3-dimensional ball containing $n$ arcs whose endpoints lie in the spherical boundary of the ball. Tangles have had various applications in knot theory and biology. Notably, tangles have been used to model enzyme actions on DNA \cite{ernst1990calculus,stolz2017pathways,buck2007predicting}. This paper presents an application of tangle calculus to the study of protein knotting. Theories that have been proposed as to how knots form in proteins usually involve the process where a terminus passes through one or two loops \cite{taylor2007protein,flapan2019topological}. Such a phenomenon can be described more precisely in terms of tangles and tangle replacements. Loosely speaking, tangles are building blocks of entanglements that can be assembled to form a knot.
 We remark that our tangle model will provide information about the overall entanglement of protein knots, but not the geometry such as length and curvature. To determine the geometry of protein knots, the combination of our model with other techniques is needed. 

We set up equations involving 2-string and 3-string tangles which correspond to a topological description of knot folding that generalize the theory of Taylor \cite{taylor2007protein} and Flapan et al. \cite{flapan2019topological} respectively. The solutions to these equations give the knot types that may arise in proteins in the future. In the process, we also describe knotoid types that arise in intermediate stages of various folding pathways. When applicable, we pay close attention to the signs of tangle replacements in an attempt to understand why the right handed version of knots or a knotoids appear more in the PDB.

\subsection*{Organization}

The paper is organized as follows. In Section \ref{Section:Taylor}, we discuss Taylor's model and its relationship to replacements of 2-string tangles. In Section \ref{section:Flapan}, we discuss Flapan-He-Wong's model and its relationship to replacements of 3-string tangles and 3-braids. In Section \ref{section:orientedmoves}, we stay oriented tangle replacements that are relevant to protein folding. In Section \ref{section:moregeneral3-stringmodels}, we generalize the models from Section Section \ref{section:Flapan} even further to study possible folding pathways of the granny knot.

\subsection*{Acknowledgements}
Research conducted for this paper is supported by the Pacific Institute
for the Mathematical Sciences (PIMS). The research and findings may not reflect
those of the Institute.
\section{Taylor's twisted hairpin pathway modeled using 2-string tangles}\label{Section:Taylor}

Consider a ball around a portion of the knot so that the sphere boundary meets the knot in an even number of points. Such a ball containing a finite number of strings is a tangle. The process of removing a tangle from the knot, and replacing it with a new tangle is called a \textit{tangle replacement}. 
A terminus threading through a loop can be modeled as a tangle replacement 
involving a crossing change 
as shown in Fig. \ref{fig:replacements}.  Taylor's twisted hairpin theory is illustrated in Fig. \ref{fig:replacements}a where a loop is first formed and then a terminus passes through this loop.  In order to identify the knot type, the ends must be connected   (Fig. \ref{fig:replacements}b).   There are many ways to connect the termini to obtain the same knot type, but we can connect the endpoints in such a way that the tangle replacement is visible. 
The yellow balls in Fig. \ref{fig:replacements} are 2-string tangles.  They indicate the difference between a terminus passing through the loop and forming a knot versus a terminus not  passing through the loop.  Note that this difference is equivalent to changing a crossing:  which string crosses over the other string switches in a crossing change. 
 If the terminus does not go through the loop, we do not obtain a knotted protein as shown on the left of Fig. \ref{fig:replacements}b, while passing through the loop results in a knot as shown on the right of Fig. \ref{fig:replacements}b.
 We can generalize Taylor's theory by not requiring the twisted hairpin. In this case, the tangle outside of the crossing change is represented by an unknown tangle {\bf T} as shown in Fig. \ref{fig:replacements}c. Similar to Taylor's twisted hairpin pathway, we assume the protein is unknotted before the crossing change per the tangle equation on the right in Fig. \ref{fig:replacements}c.  After the crossing change, a knot {\bf K} is formed per the tangle equation on the left in Fig. \ref{fig:replacements}c. By solving the system of tangle equations in Fig. \ref{fig:replacements}c, we can predict what knot types are likely to occur in proteins.  

One solution to this system of equations is given in  Fig. \ref{fig:replacements}d. This solution is topologically equivalent to the configuration in Fig. \ref{fig:replacements}b. In topology, we can deform an object as long as we do not break/tear or glue.  In both b and d, the configuration on the left is unknotted, while that on the right is the knot $3_1$.  To see that the tangles are equivalent, one can first rotate the configurations in  Fig. \ref{fig:replacements}b by 45 degrees counterclockwise so that the yellow tangles agree with those in Fig. \ref{fig:replacements}d.  More work is needed to convert the black strings outside of the yellow tangle so that they are similar in both figures.  Note the black strings outside the yellow tangle are represented by the tangle {\bf T} in  Fig. \ref{fig:replacements}c. It is difficult to see the twisted hairpin in Fig. \ref{fig:replacements}d.  However, solutions to these tangle equations explain why only certain knots are observed in proteins.  They can also be used as a starting point to determine biologically relevent geometric configurations.

\begin{figure}[!ht]
 \centering
a.) \includegraphics[width=0.33\textwidth]{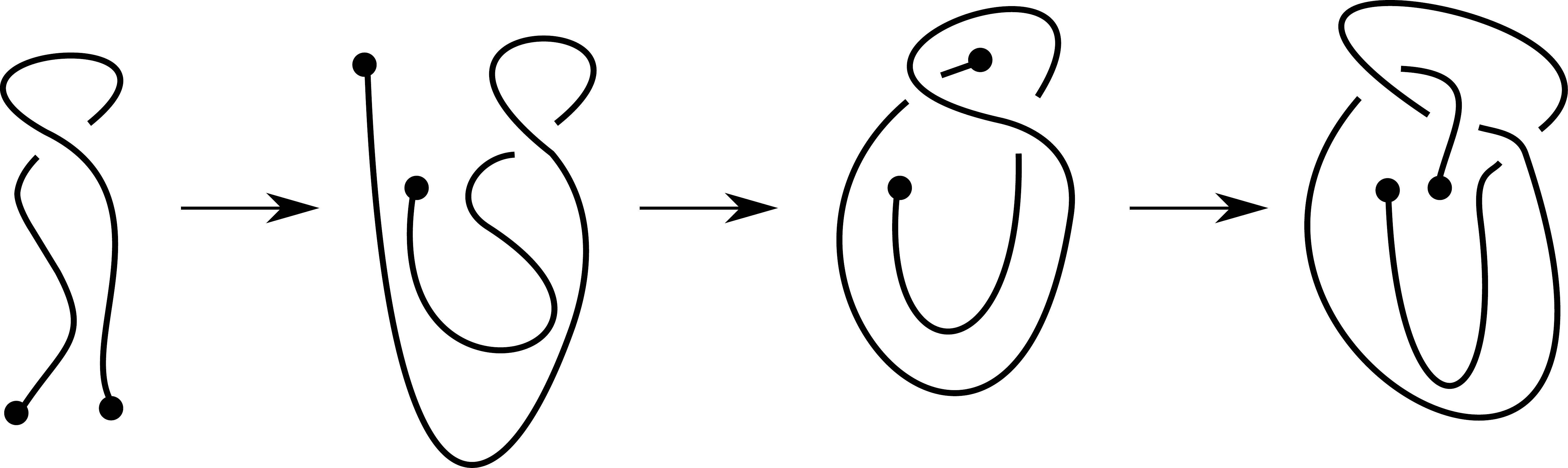}
c.) \includegraphics[width=0.5\textwidth]{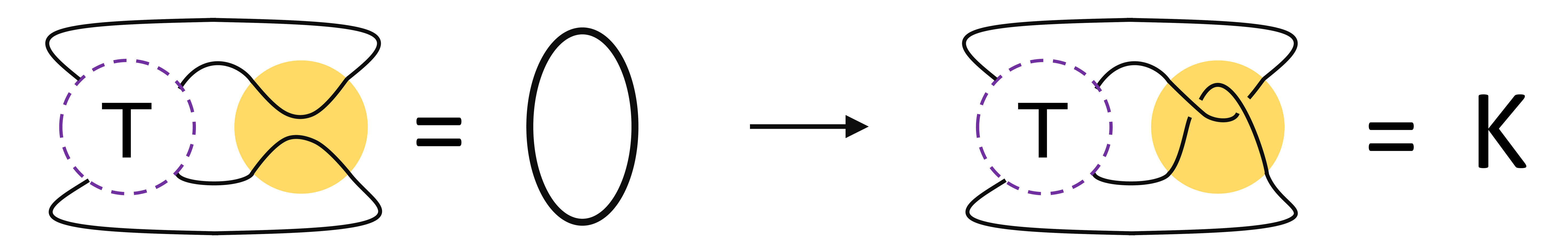}  
b.)\includegraphics[width=0.23\textwidth]{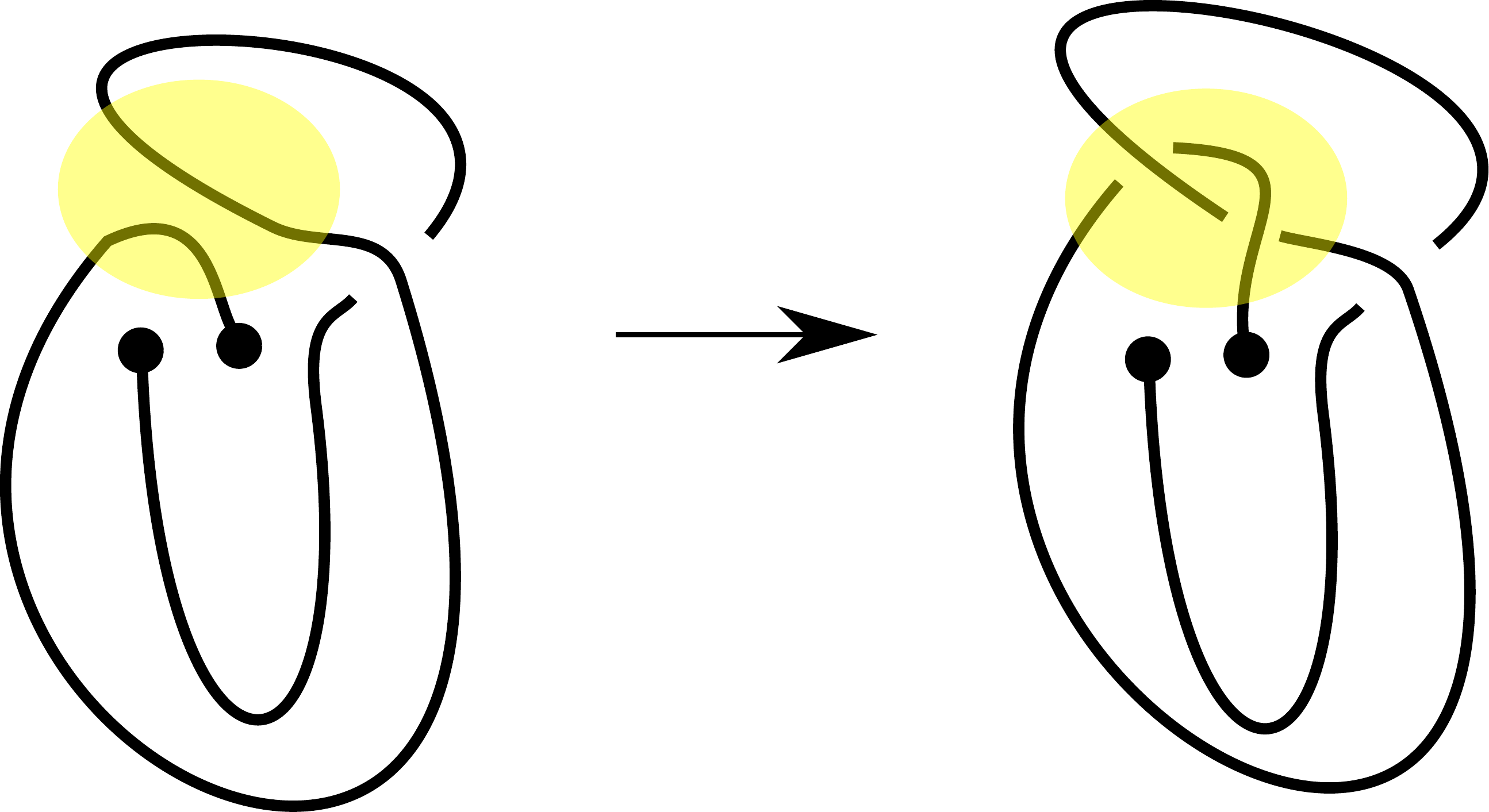} 
d.) \includegraphics[width=0.53\textwidth]{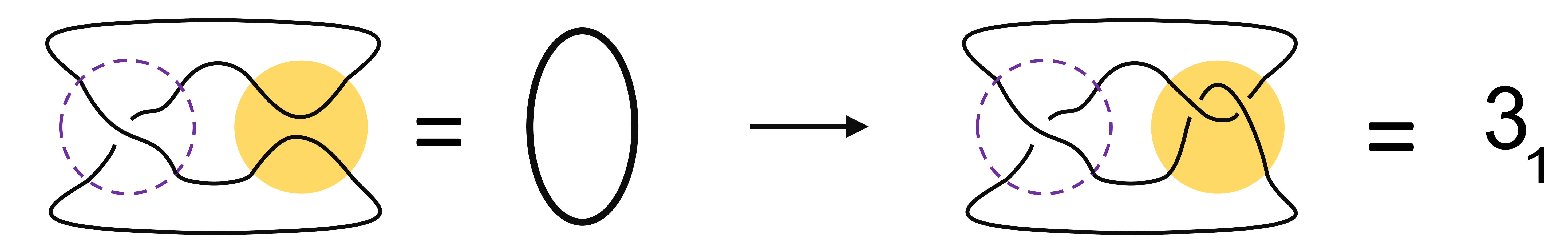}  

e.) \includegraphics[width=0.53\textwidth]{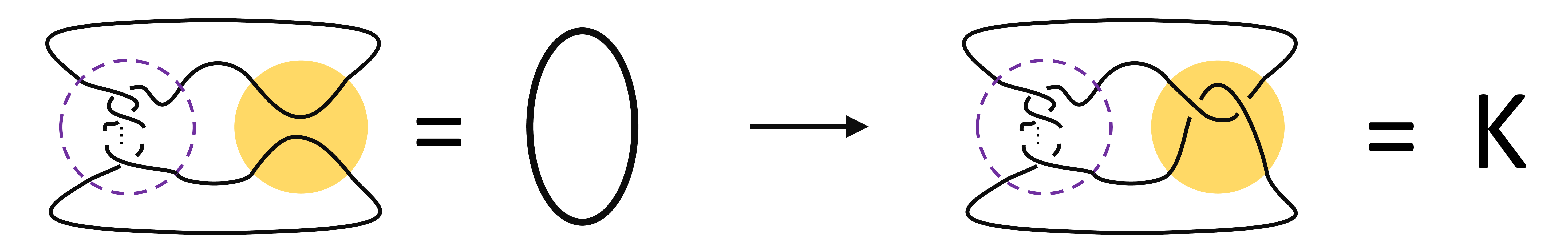} 
        \caption{
 Interpreting Taylor's twisted hairpin theory as a 2-string tangle replacement called a crossing change.  The protein backbone is represented by the black curve. 
(a)   In Taylor's twisted hairpin theory, a loop is first formed and then a terminus threads through the eye of the hairpin.
  (b) The ends of the protein are connected. The left-hand side represents the unknotted configuration where the terminus does not pass through the loop.  The right-hand side represents the knotted configuration obtained when the terminus passes through the loop.  The yellow balls containing two black segments indicate the difference between these two configurations.
(c)  System of tangle equations modeling Taylor's twisted hairpin theory.  Passing through the loop is modeled by the crossing change occuring in the  yellow tangles.  The configuration outside of this is modeled by the unknown tangle {\bf T}.   The knot type obtained by passing through the loop is indicated by {\bf K}.
(d)  A solution to the tangle equation in (c) where the tangle {\bf T} contains one crossing and the knot {\bf K} equals $+3_1$.  This solution is topologically equivalent to (b).
(e) Solutions to the tangle equation in (c) assuming the Taylor's twisted hairpin model.  The knot on the right belongs to the twist knot family.}
        \label{fig:replacements}
\end{figure}

\section*{Topology versus Geometry: Chirality implications}

  Proteins have a strong chiral bias towards right-handedness.  Many motifs including beta sheets and alpha helices twist in a right-handed fashion. Very few motifs contain left-handed twists. Thus, knots containing right-handed twists would be more expected. Assuming a right-handed chirality bias, the most likely knots to occur in proteins via Taylor's twisted hairpin model are the knots $+3_1$, $4_1$, and $-5_2$. The majority of protein knots are trefoil knots.  But the number of positive trefoils far outweigh the number of negative trefoils.  A chirality bias towards right-handedness would explain this as the  $+3_1$ contains three right-handed crossings while the $-3_1$ contains three left-handed crossings.  The knot $4_1$ is achiral (equivalent to its mirror image).  Its minimal diagram (the 2-dimensional projection with the minimum number of crossings) contains two right-handed crossings and two left-handed crossings.  But a non-minimal diagram of $4_1$ exists containing 5 right-handed crossings. The minimal diagram knot $-5_2$ contains three right-handed crossings and two left-handed crossings, but one can also draw this knot using six right-handed crossings.  The knots  $4_1$ and $-5_2$  are both twist knots containing right-handed crossings that are trapped via a 2 crossing left-handed clasp.  The $+6_1$ knot is the exception in that 4 out of 6 crossings in its minimal diagram are left-handed. An alternative folding pathway which includes right-handed twists has been proposed for the protein DehI, which contains this knot. This pathway is modeled below using 3-string tangles.

To capture this right-handed chirality bias, we focus on the tangle model where the zero crossing tangle (modeling the unknotted protein before passing through the loops) is replaced with a two crossing tangle containing right-handed twists.  For example, the positive trefoil is obtained by trapping a single right-handed crossing before passing through the loop while two left-handed crossings are needed to obtain the negative trefoil. Thus one should expect more positive than negative trefoils, as is the case.

Taylor's twisted hairpin model is shown as a tangle equation in Fig. \ref{fig:replacements}e. In this case, the tangle {\bf T} consists of vertical crossings.  The knot $K$ obtained by passing the endpoint through the loop must be a twist knot.
 We will first prove that if the tangle {\bf T} is simple, then Taylor's twisted hairpin model is the only solution to the system of tangle equations in Fig. \ref{fig:replacements}c.  %We also classify the likely knot types that can result from a single crossing change.
 The simplest tangles are rational tangles, which we define in the next section.  While we give the full mathematical definition, only the pictorial representations are needed to understand the results.

%The solutions to the tangle equations in Fig. \ref{fig:replacements}c are known \cite{}.  As discussed in supplementary material these equations can be translated into algebraic equations to determine the solutions.  However all solutions are pictorial

%In our setting, the threading process is treated as a tangle replacement. Roughly, we interpret Taylor's protein folding theory, which involves the threading of a terminus through one loop, as a 2-string tangle equation involving a crossing change illustrated in Fig. \ref{fig:replacements}(c). The solution to this tangle equation explains why only twist knots have been observed in proteins. 

 % In particular we will show that only twist knots are the {The equation that models the threading of a terminus through one loop, generalizing Taylor's folding pathways} The goal is to determine all likely knots K and configurations T

\subsection{Rational 2-string tangles}

In order to gain a better understanding of the significance of knotting on a molecular level, it is of interest for chemists to synthesize topologically complex molecular knots. An example of artificially engineered knotted proteins is a long polymer created by T. Sayre et al. \cite{sayre2011protein}. This protein is constructed by polymerizing the carbonic anhydrase II protein containing the $+3_1$ knot.

One successful strategy of creating molecular knots, which is very closely related to a mathematical definition of a rational tangle, was employed by Dietrich-Buchecker and Sauvage \cite{dietrich1989synthetic}, where they intertwined two molecular threads together around a central transition metal, resulting in horizontal twists  (see Fig. \ref{fig:sauvage}). After linking the ends of the ligands in the helicate together, the transition metals can be removed to obtain a trefoil knot. This process is strikingly reminiscent of a more general construction of a family of tangles familiar to knot theorists called \textit{rational tangles}.

\begin{figure}[!ht]
  \centering
    \includegraphics[width=0.5\textwidth]{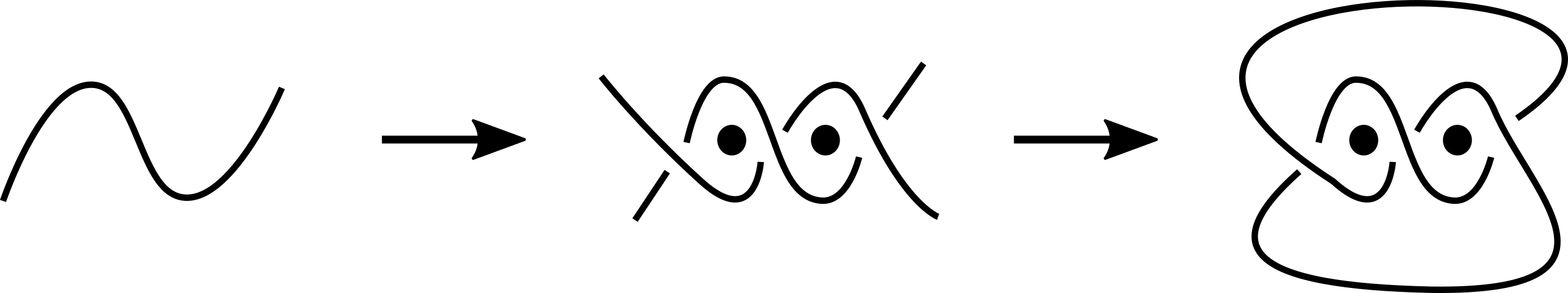}
        \caption{Using two central transition metals to synthesize a molecular trefoil knot.}
        \label{fig:sauvage}
\end{figure}

For each 2-string tangle, we can assign labels NE, NW, SW, and SE to the endpoints of the strings as shown in Fig. \ref{fig:rational} (a). A 2-string tangle is said to be \textit{rational} if it can be created by starting with a zero crossing tangle (see Fig. \ref{fig:rational}) (a),   and alternating between twisting the NE and SE endpoints around each other and twisting the SW and SE around each other.  We will also designate a plus or a minus sign to a half-twist depending upon whether the understrand at the crossing has a positive or a negative slope, respectively.
 With these conventions, we can represent each rational tangle as a vector $(a_1,...,a_n)$ where each entry is an integer corresponding to the number of signed twists of the NE and SE or SW and SE endpoints recorded in order. We additionally require that the last entry $a_n$ in the vector represents the half twists of the NE and SE endpoints even if that number is 0. For instance, the third step of Fig. \ref{fig:rational} (a) has 0 as the last entry because we end by twisting the SW and SE endpoints.   Thus since we start with the 0-tangle, $n$ will be odd and $a_i$ will correspond to horizontal twists if $i$ is odd and vertical twists if $i$ is negative.
Horizontal twists will be righthanded if $a_i >1$ and lefthanded if $a_i < 1$, while vertical twists  will be righthanded if $a_i <1$ and lefthanded if $a_i > 1$,
From a vector $(a_1,...,a_n)$ representing a rational tangle, we can form a fraction $\dfrac{p}{q}$ using the following correspondence:
%\begin{align*}
    $\dfrac{p}{q}=a_n+\frac{1}{a_{n-1}+ \cdots \frac{1}{a_2+\frac{1}{a_1}}}$.
%\end{align*}
 John H. Conway proved a powerful classification result, which states that two rational tangles are equivalent if and only if their corresponding fractions are equal \cite{conway1970enumeration}. 

   We can create a knot from a tangle. The \textit{numerator closure} of a tangle $A$, denoted $N(A)$ is formed by connecting the NW and NE endpoints and the SW and SE endpoints by arcs as shown in Fig. \ref{fig:rational} (b). When a knot $N(A)$ is obtained as the numerator closure of a rational tangle $A = (a_1,...,a_n)$, we call $N(A)$ a \textit{rational knot} (also referred to as 2-bridge or 4-plat in mathematical literature).

\begin{figure}[!ht]
  \centering
    \includegraphics[width=0.7\textwidth]{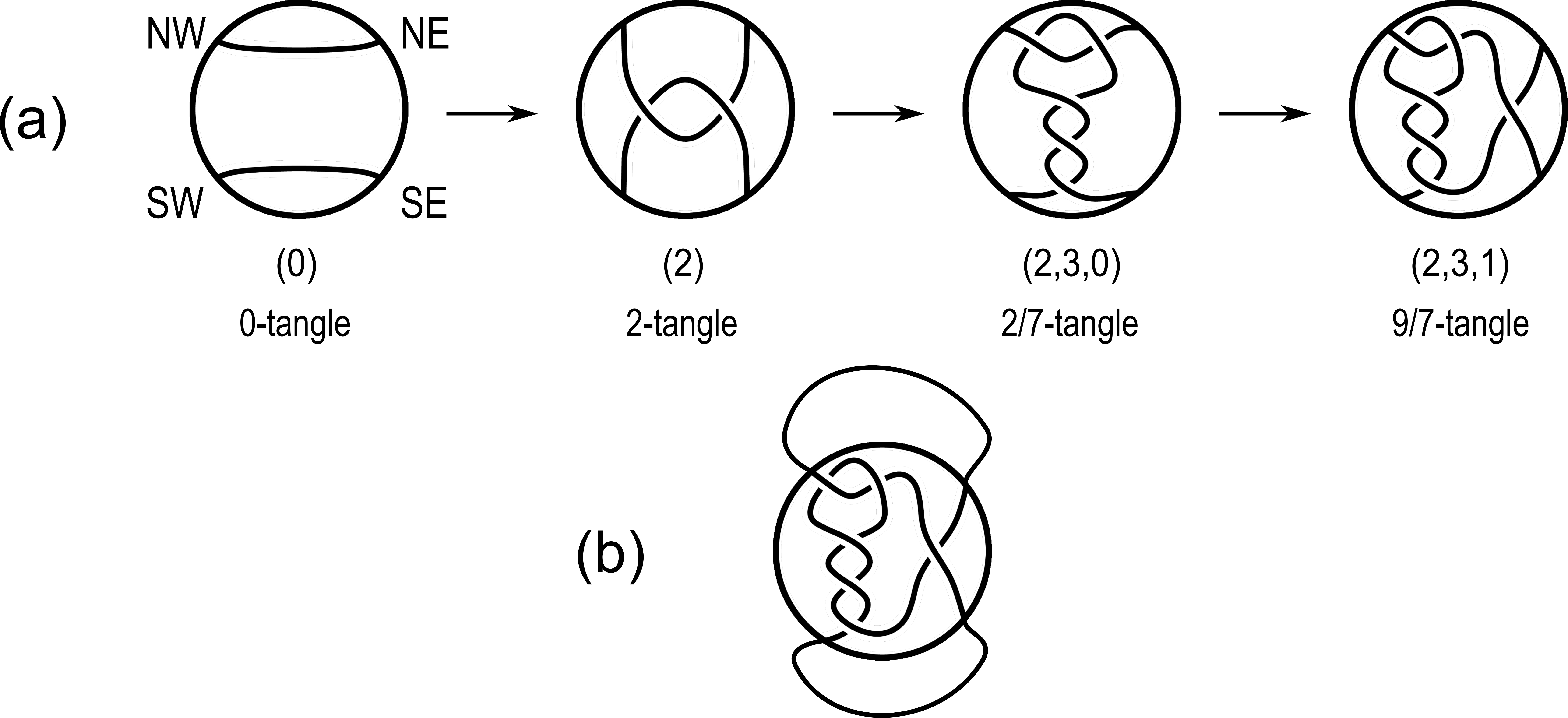}
        \caption{(a) The process of creating the 30/7-tangle. (b) The numerator closure of the 9/7 tangle.}
        \label{fig:rational}
\end{figure}

Rational tangle theory provides a strategic blueprint for constructing more complicated knots in the future. As a consequence of our tangle analysis, we will see that even if we would like to construct a rational knot by starting with a rational tangle, and letting one endpoint goes through a loop, then the resulting knot must be a twist knot, which is a very simple type of rational knot.

\subsection{Taylor's twisted hairpin pathway and 2-string tangles}

In Taylor's theory, he assumes that the unknotted protein takes the shape of a twisted hairpin \cite{taylor2007protein}. This configuration is a particular instance of the unknot in the form of the numerator closure of the sum of a simple rational tangle with the 0-tangle. We can generalize Taylor's theory by assuming that we start out with the sum of \textit{any} rational tangle with the 0-tangle whose closure is the unknot. We are interested in what sort of knots can result by replacing the 0-tangle with the $+2$ tangle because the replacement represents the terminal end passing through the loop. This comes down to solving a system of two tangle equations illustrated in Fig. \ref{fig:replacements}.

Twist knots are the only knot types discovered in proteins. This is clearly true if the protein assumes the form of a twisted bobby pin with a loop on top before a terminus threads through the loop. The solution to the equation in Fig. \ref{fig:replacements} shown in the following proposition implies that we will still only discover twist knots even if we start out with the sum of an arbitrary rational tangle with the 0-tangle whose closure is the unknot. 
\begin{res}\label{prop:2stringeq}
Suppose that $T$ is a rational tangle represented by $\dfrac{p}{q}$. Given the following system of tangle equations:
\begin{align}
N\left(T+0\right) &= N\left(\frac{1}{0}\right) \label{eq1}\\
N\left(T+2\right) &= K \label{eq2},
\end{align}
then $K$ is a twist knot and the tangle $T$ consists of vertical twists as shown in Fig. \ref{fig:replacements}e.
\end{res}
\begin{proof}
Equations \ref{eq1} and \ref{eq2} can be rewritten as 
\begin{align}
N\left(\frac{p}{q}\right) &= N\left(\frac{1}{0}\right) \label{eq3}\\
N\left(\frac{p}{q}+2\right) &= K \label{eq4}.
\end{align}

Due to Schubert's classification \cite{schubert1956knoten}, Equation \ref{eq3} implies that since $p = 1$. Substituting $p=1$ into \ref{eq4} gives

\begin{align*}
K = N\left(\frac{1}{q}+2\right).
\end{align*}
These are twist knots, which can be visualized as the numerator closure of the sum of two tangles, as shown in Fig. \ref{fig:proteinknots} (g), and the tangle $T$ is made up of vertical twists.
\end{proof}

Twist knots (see Fig. \ref{fig:proteinknots} (g)) were the only knot type that resulted from the model above. It assumes simplicity by using rational tangles. We will show in the next subsection that the requirement that $T$ has to be rational could be relaxed, but it would lead to a more complicated situation that might not be expected frequently in nature. For an example of a tangle that is not rational, the readers may consult Fig. \ref{fig:nontwist}.

\subsection{Non-twist knots in proteins}

Even though only twist knots have been identified on protein backbones, it is natural to ask how plausible it is for knotted proteins to contain more complicated knots. We can find all solutions to the tangle equations in Fig \ref{fig:replacements}c when the knot resulting from a single crossing change is a rational knot \cite{darcy2005solving, darcy2006topoice}. See the appendix or the TopoICE-X software for these solutions.  

The smallest crossing knot that has not been recorded in the PDB (even though it is predicted by the computer artificial intelligence system AlphaFold) is the knot $5_1$. Another possible reason why protein knots that are not twist knots have not been observed as frequently is the following. From Fig. \ref{fig:replacements}, any protein knot can be untied by switching only one crossing if we adhere to Taylor's theory. In particular, all twist knots satisfy this condition. Using techniques from geometric topology, one can show that it is not possible to transform the $+5_1$ knot into the unknot by switching only one crossing \cite{darcy1998applications}. In other words, there is no solution to the tangle equation in Fig. \ref{fig:replacements}c when $K$ is the knot $5_1$. Rather, one needs to perform at least \textit{two} crossing switches to untie the $+5_1$ knot. Therefore, $+5_1$ can never arise in nature via Taylor's theory.

The next 2 simplest knots that have not yet been found in the PDB are the six crossing knots  $6_2$ and $6_3$. Both of these knots can be obtained from the unknot via a single crossing change per Fig. \ref{fig:nontwist}. However note that the solution for $T$ when $K$ is one of these non-twist six crossing knots is rather complicated. For any knot $K$ that is not a twist knot, the solution for $T$ will be at least as complicated as those shown in Fig. \ref{fig:nontwist}. We remark that in this illustrations, instead of replacing the 0-tangle with the +2 tangle, we replace the $+1$ tangle with the $-1$ tangle. Readers can see that doing this new type of replacement accomplishes the crossing change task as well.

\begin{figure}[!ht]
  \centering
      \centering
     \begin{subfigure}[b]{0.3\textwidth}
         \centering
         \includegraphics[width=\textwidth]{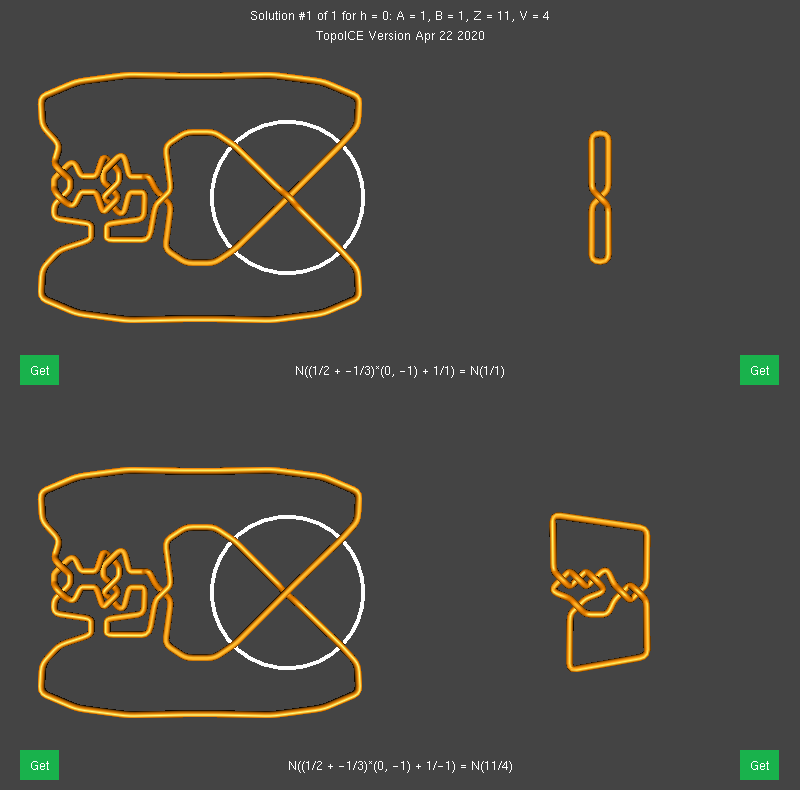}
         \caption{}
         \label{fig:62solution}
     \end{subfigure}\\
          \begin{subfigure}[b]{0.3\textwidth}
         \centering
         \includegraphics[width=\textwidth]{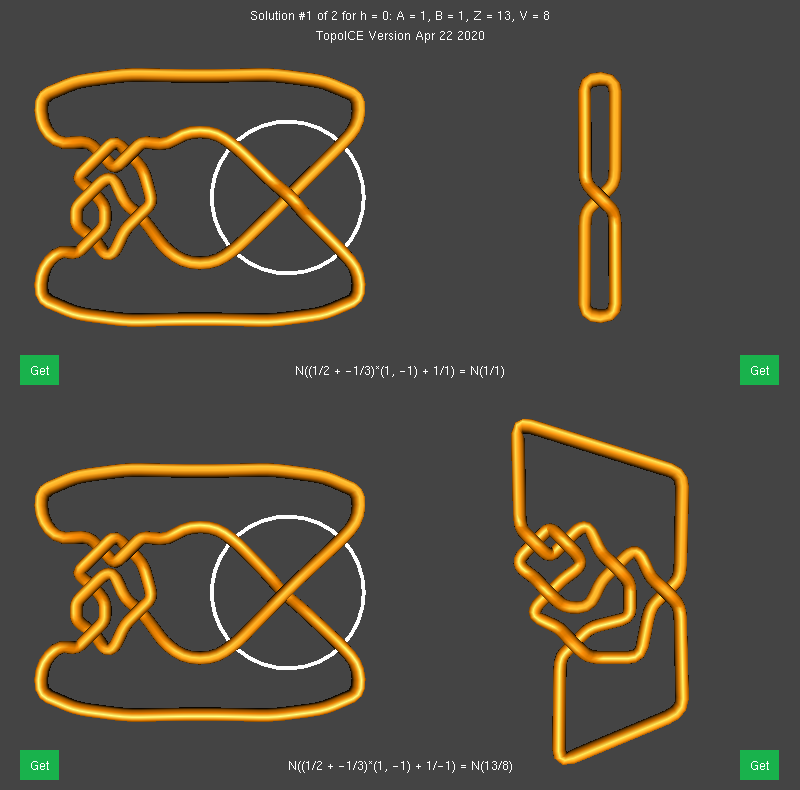}
         \caption{}
         \label{fig:63solution}
     \end{subfigure}
               \begin{subfigure}[b]{0.3\textwidth}
         \centering
         \includegraphics[width=\textwidth]{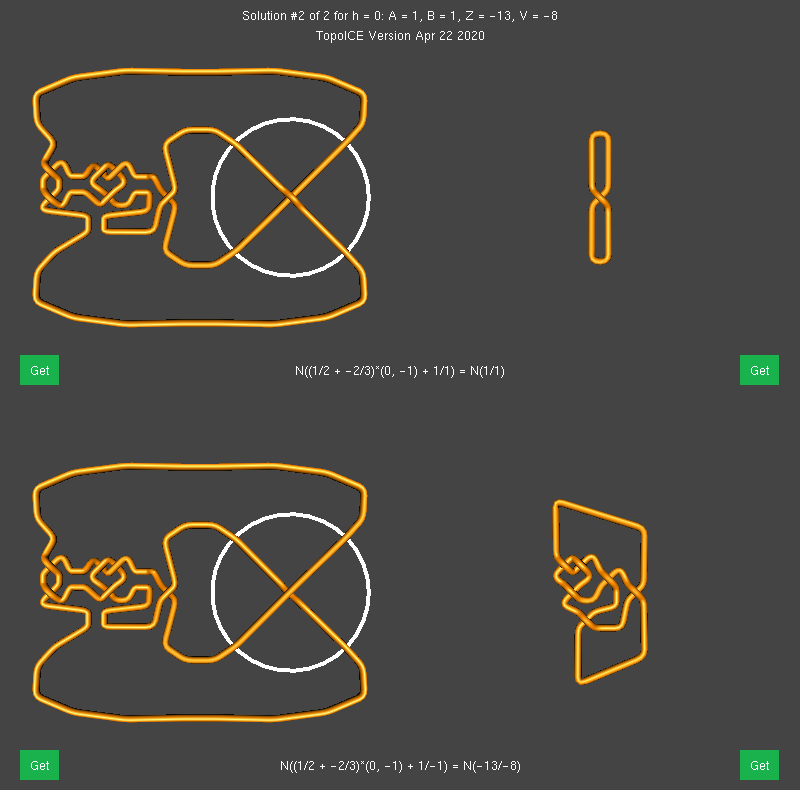}
         \caption{}
         \label{fig:63solution2}
     \end{subfigure}
        \caption{
(a) Obtaining the non-twist knot $6_2$ via a single crossing change. Subfigures (b) and (c) depict two solutions for the tangle $T$ in Fig. \ref{fig:replacements}c when $K$ is the non-twist knot $6_3$.
These illustration were created using the software TopoICE-R \cite{darcy2006topoice}.
}
\label{fig:nontwist}
\end{figure}

\subsection{Knotoids and generalized Taylor's model}
A \textit{knotoid diagram} is an immersion of a closed interval into the sphere, where each double point is decorated with an over-crossing/under-crossing information. A \textit{knotoid} is an equivalence class of knotoid diagrams modulo classical Reidemeister moves applied away from the endpoints. These Reidemeister moves are depicted in Fig. \ref{fig:Reid}. Knotoids have been successfully used by various authors to model knotted proteins. In particular, folding pathways of proteins can be analyzed as local changes in knotoid diagrams. For instance, Barbensi and Goundaroulis studied a local move where an open end of the protein is tucked under or over another strand \cite{barbensi2021f}. Such a move is called a \textit{forbidden move}, and they are depicted in Fig. \ref{fig:signedforbidden}. We remark that while a knotoid looks like a tangle, they are not the same object. For instance, the third picture the sequence in Fig. \ref{fig:replacements} (a) is trivial as a tangle, but it is nontrivial as a knotoid.

\begin{figure}[ht!]
\centering
\includegraphics[width=10cm]{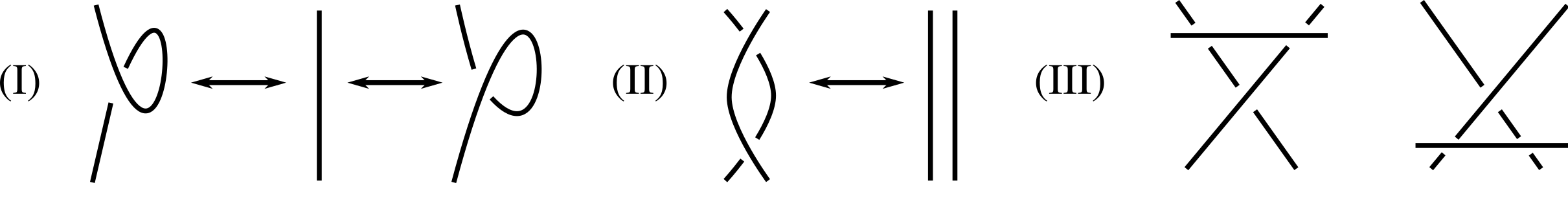}
\caption{Reidemeister moves.}\label{fig:Reid}
\end{figure}

\begin{figure}[ht!]
\centering
\includegraphics[width=5cm]{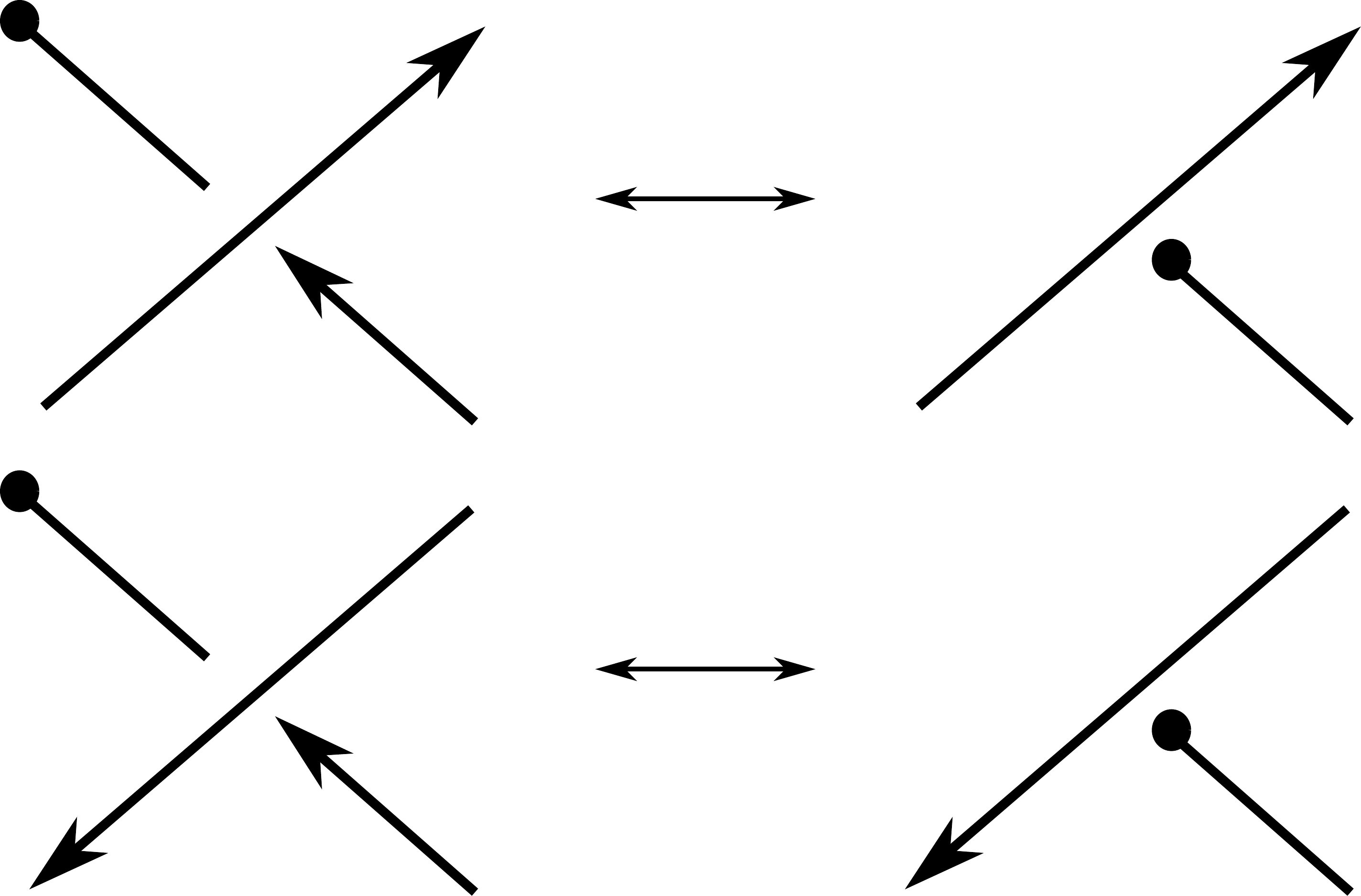}
\caption{(Top) A positive forbidden move. (Bottom) A negative forbidden move.}\label{fig:signedforbidden}
\end{figure}

Although non-twist knots have not been recorded in the PDB, they are predicted by AlphaFold. Thus, it may be useful to analyze more general tangle models. Generalizing rational tangles, a collection of rational tangles can be arranged into a necklace-like formation to form a new class of knots called \textit{Montesinos knots} and two arcs can be removed to obtain \textit{Montesinos tangles} as demonstrated in Fig. \ref{fig:Montesinos2}. In the result below, the tangle $S$ and $S'$ are tangles that differ by a \textit{2-move} as shown in Figure \ref{fig:2move}. Performing a 2-move achieves the loop threading action in Taylor's model.

\begin{figure}[ht!]
\centering
\includegraphics[width=10cm]{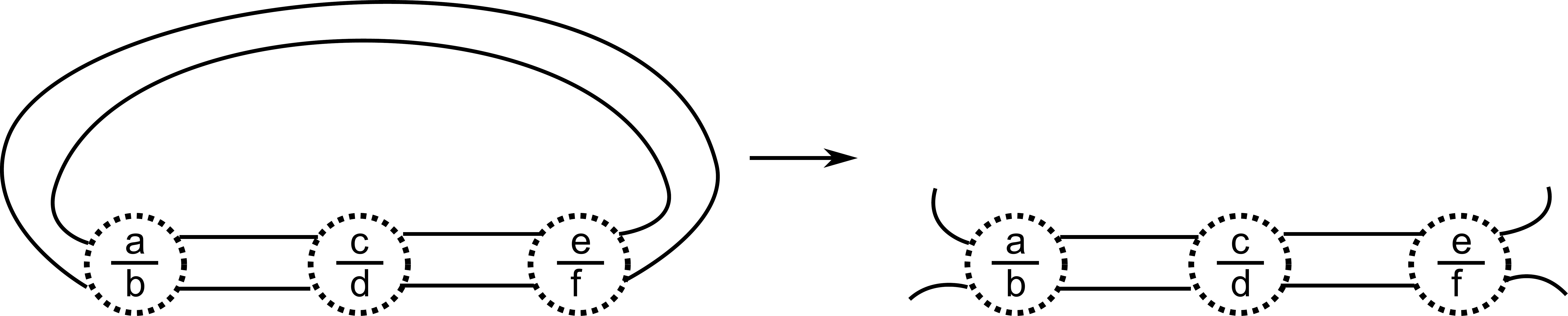}
\caption{At left, a Montesinos knot of length three with vector $(a/b,c/d,e/f)$. At right, some arcs are removed from a Montesinos knot to get a Montesinos tangle $(a/b,c/d,e/f)$. In general, a Montesinos knot of length $n$ admits a diagram that can be visualized as a necklace on $n$ beads, where each bead is a rational tangle.}\label{fig:Montesinos2}
\end{figure}

\begin{figure}[ht!]
\centering
\includegraphics[width=7cm]{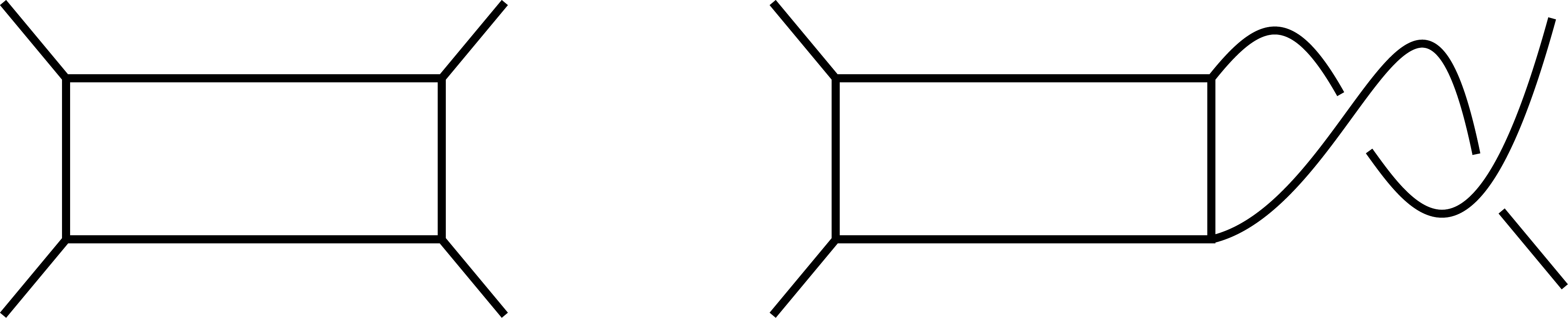}
\caption{At left, the tangle $S$ appearing in Result \ref{res:Mon}. At right, the tangle $S'$ appearing in Result \ref{res:Mon}. The boxes contain identical subtangles.}\label{fig:2move}
\end{figure}

\begin{res}
\label{res:Mon}
Suppose that $T$ is a Montesinos tangle. Given the following system of tangle equations:
\begin{align}
N\left(T+S\right) &= Unknot \label{eq2.1}\\
N\left(T+S'\right) &= K \label{eq2.2},
\end{align}
then $K$ is a 4-plat knot admitting a Montesinos diagram $(p_1/q_1,p_2/q_2,a)$, where $p_1q_2 + p_2q_1 \cong \pm 1 \mod q_1q_2$ and $a$ is an integer determined by $T$. Furthermore, suppose that $K$ has the bracket polynomial $V_K(A)$, then there is a nontrivial knotoid $k$ in the folding pathway, whose bracket polynomial is $V_k(A)=(A^{\pm 3}V_K(A)+A^{\mp 3})/(A^3+A^{-3})$.
\end{res}
\begin{proof}
This was pointed out as Theorem 4.9 by Nogueira and Salgueiro \cite{nogueira2022embeddings}. To summarize, since $T$ is Montesinos, the tangle $S$ must be integral. The knot $N(T+S)$ has an associated 3-manifold coming from a 2-fold branched cover, which is a Seifert fibered space $M$. The 3-manifold $M$ is also parametrized by a vector of rational numbers $(p_1/q_1,\cdots,p_n/q_n)$ where each of these rational numbers are extracted directly from $T$ and $S$. But since $N(T+S)$ is the unknot, its branched cover is the 3-sphere. This implies that $n=2, q=1$ and $p_1q_2 + p_2q_1 \cong \pm 1 \mod q_1q_2$ by the classification of Seifert fibered spaces. It is a standard fact that a Montesinos knot of length 3, where one of the parameters is integral, is a 4-plat knot.

One can verify that the bracket polynomial of $k$ coincides with the bracket polynomial of the virtual closure of $k$. Here, recall that the virtual closure is obtained by connecting up the open ends of the knotoid and declaring any crossings encountered by the closure arc virtual crossings. The polynomial $V_k(A)$ is obtainable via the skein relation proved in Theorem 1 of \cite{kamada2002virtualized}.
\end{proof}

\begin{figure}[ht!]
\centering
\includegraphics[width=7cm]{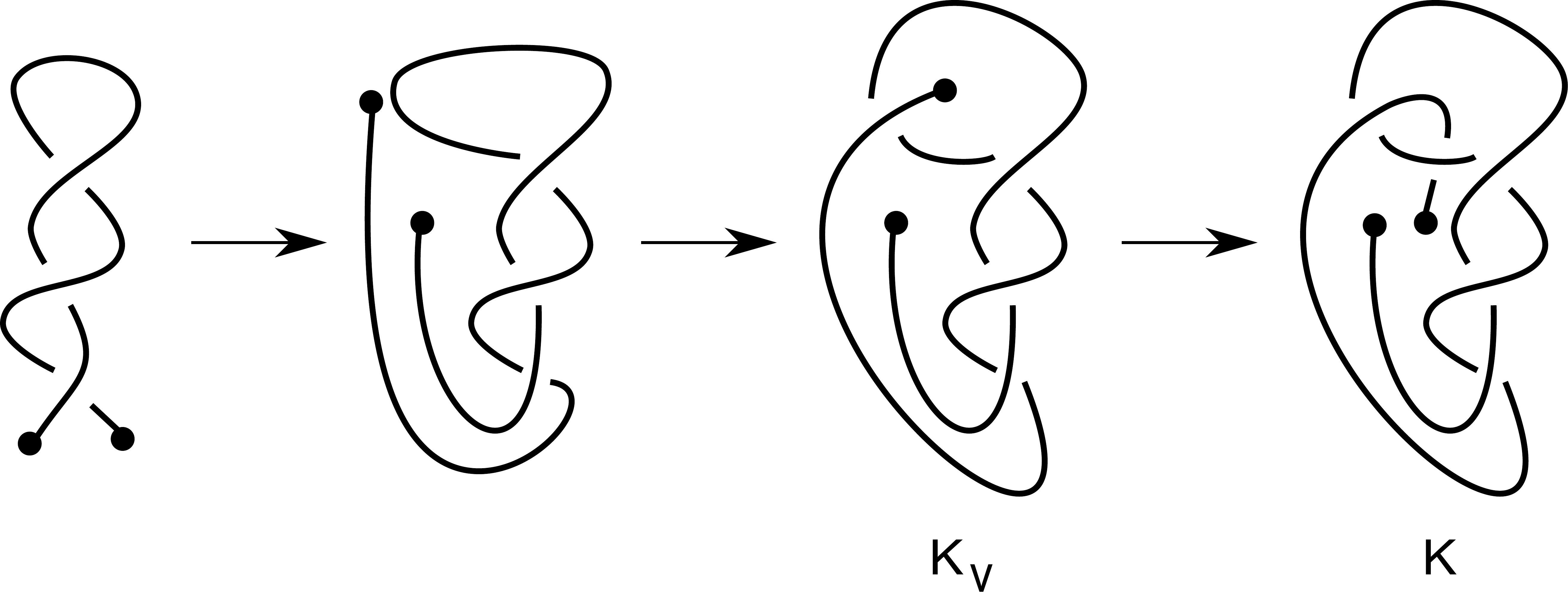}
\caption{If one treats each stage of Taylor's twisted hairpin pathway as a knot, then all the knots may be trivial except the final stage. Treating each stage as a knotoid, however, we always get another nontrivial knotted object $K_v$ in an intermediate stage.}\label{fig:virtualknotpath}
\end{figure}

The result above gives infinitely many knotoids with forbidden number one. Furthermore, since 4-plat knots are alternating, we can use Corollary 1.4 of Kamada \cite{kamada2004span} to say that it is impossible to perform Reidemeister moves to make the endpoints of the knotoid $K_v$ lie in the same region.
\section{Flapan-He-Wong's theory}\label{section:Flapan}
While it is possible that Taylor's theory gives the primary folding pathway for protein knots, several computational and experimental results hint at the existence of other folding mechanisms \cite{andersson2011effect,andersson2009untangling,lee2017entropic,lou2016knotted,zhang2016characterization,bolinger2010stevedore}. In particular, these studies show that a folding pathway may contain an intermediate stage where a nontrivial knot is created before the final threading of a terminus through a loop. Such a phenomenon never occurs in the twisted hairpin pathway as the twisted hairpin configuration is always topologically unknotted before a terminus goes through the eye of the hairpin.

The knotted DehI protein is a notable example of a knotted protein whose folding pathways may involve a knotted intermediate. Fig. \ref{fig:fold} illustrates the two parallel pathways to the $+6_1$ knot found by using molecular dynamics simulations with a coarse-grained G\={o} model of the folding of DehI. We remark that what we refer to as the green and the red loop were also called the B-loop and the S-loop respectively in B{\"o}linger et al. \cite{bolinger2010stevedore} In steps a and b, the green loop and the red loop are created. In step c, one more twist is added to the red loop, after which two parallel pathways are possible. In the first possibility, the green loop flips over the red loop and the light blue terminus. The light blue terminus then threads through the red loop yielding the $+6_1$ knot in step d. In the alternate pathway, the blue terminus threads through the red loop before the green loop flips over both the red loop and the blue terminus. Observe that the two pathways are distinct since the knot present in step c involving the light blue terminus is the $4_1$ knot while the knot present in step c involving the (darker) blue terminus is the unknot.

The simulations for DehI outlined above motivated Flapan et al. \cite{flapan2019topological} to consider more general pathways that involve loop-flipping. We now explain Flapan et al.'s topological descriptions of protein folding. In step c of Fig. \ref{fig:fold}, the red loop contains two positive twists and the green loop contains one negative twist. Flapan et al.'s relaxed these restrictions on the number of twists and allowed zero, one, or two twists of either sign in each of the loops. Next, notice that in step c of Fig. \ref{fig:fold}, the segment between the blue and the green passes behind the red arc. On the other hand, Flapan et al. allowed the segment between the blue and green segments to either pass behind or in front of the red arc. After this point, two parallel pathways are created depending on whether the green loop flipping or the threading of the blue terminus through the red loop happens first, just like in B{\"o}linger et al.'s simulation. However, Flapan et al. noticed that there are four distinct ways that the blue terminus can go through the green and the red loop, which they referred to as LL, RR, LR, and RL. These four distinct ways are illustrated in Fig. 6 of \cite{flapan2019topological}.

\begin{figure}[ht!]
\centering
\includegraphics[width=8cm]{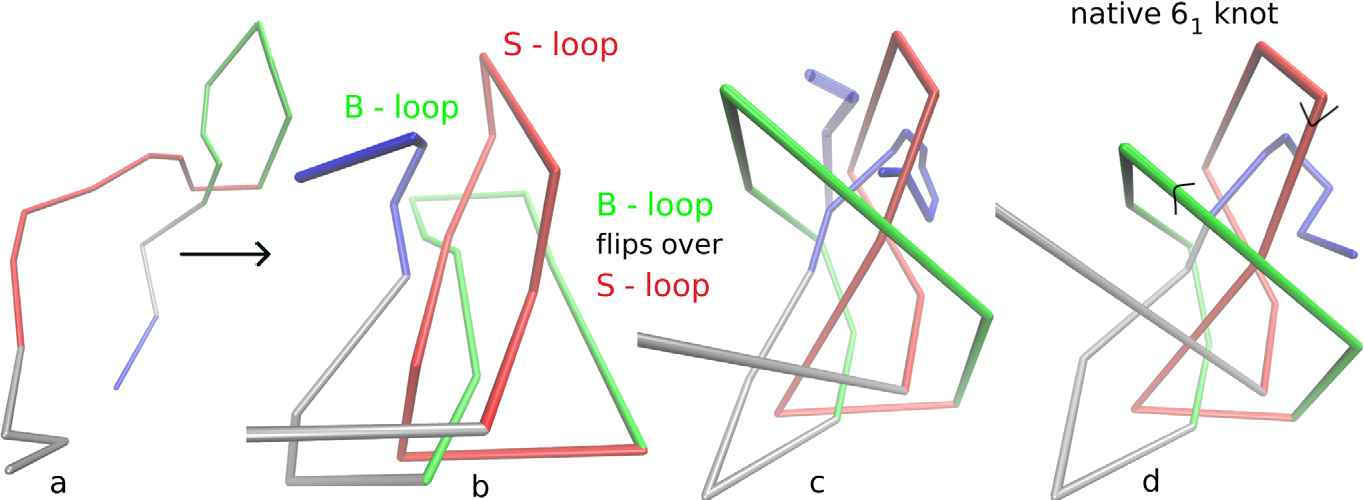}
\caption{The folding pathways of DehI. In step a, the green B-loop is created. In step b, the red S-loop is created. There are two possibilities for step c and d depending on the order of the green B-loop flipping and the threading of the blue terminus through the red S-loop. To differentiate between the two pathways, the authors color the blue terminus in a lighter shade for the first pathway. In step c of the first pathway, one more twist is added to the red S-loop. Then, the green B-loop flips over both the red S-loop and the light blue terminus. At this point, performing the probabilistic closure results in the $4_1$ knot. In step d, the light blue terminus threads through the red S-loop resulting in the $+6_1$ knot. In step c of the second pathway, one more twist is added to the red S-loop as well. Then, the blue terminus threads through the red S-loop. The green B-loop then flips over the blue terminus and the red S-loop in step d, resulting in the $+6_1$ knot. If we perform the probabilistic closure in the second pathway after right after step c, we get the unknot. Image modified from B{\"o}linger et al.'s paper \cite{bolinger2010stevedore}}\label{fig:fold}
\end{figure}
\subsection{Introducing 3-braids}
A \textit{3-braid} is defined
to be a set of three strings where an endpoint of each string is attached to a left vertical bar, and the other endpoint is attached to
a right vertical bar with the additional requirement that as we move along each string from the left bar to the right bar, the string
heads rightward. One can easily see that a 3-braid can be treated as a 3-string tangle if one treats the vertical bars as subsets of
the sphere boundary of a ball. Any 3-braid can be put into one of the forms illustrated in Fig. \ref{fig:convention} depending on whether the
number of the boxes is even or odd. In Fig. \ref{fig:convention}, the vertical bars for the braids are shown in blue, but in subsequent depictions
of 3-braids, we will omit the vertical bars. Each box represents some number of twists considered with signs (see Fig. \ref{fig:braidtwist}). When the box contains the top two strings, a positive twist corresponds to a crossing where the undercrossing has positive slope. On the other hand, if the box contains the bottom two strings, a positive twist corresponds to a crossing where
the undercrossing has negative slope. A sequence of integers $a_1,...,a_n$ determines a 3-braid and
we will use the notation $\mathcal{T}(a_1,...,a_n)$ to represent the 3-braid associated to such integers.

A braid is said to be \textit{alternating} if as one traverses the braid from left to right, the crossings alternate between over and
under. Let $E$ denote the special 3-braid $\mathcal{T} (1,-1,1)$ and $-E$ denotes the braid $\mathcal{T}(-1,1,-1)$. By $\mathcal{T}(a_1,...,a_n)+kE$, we mean
that the braid $\mathcal{T}(a_1,...,a_n)$ is followed by $k$ copies of the braid $E$ if $k$ is positive or the braid $-E$ if $k$ is negative. Cabrera-Ibarra
proved that any 3-braid can be put in standard form, which means that any generic braid  $\mathcal{T}(a_1,...,a_n)$ is equivalent to a braid
of the form $\mathcal{T}(b_1,...,b_n)+kE$ where $\mathcal{T}(b_1,...,b_n)$ is alternating and $k$ is an integer.
\begin{figure}[!ht]
  \centering
    \includegraphics[width=0.4\textwidth]{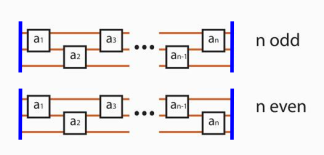}
        \caption{Any three braid has one of the following diagrams.}
        \label{fig:convention}
\end{figure}

\begin{figure}[!ht]
  \centering
    \includegraphics[width=0.6\textwidth]{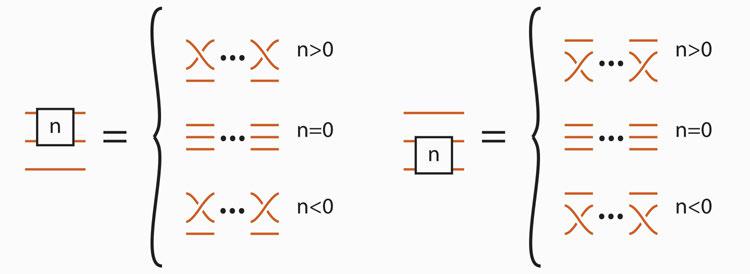}
        \caption{Any 3-braid can be obtain by stacking some number of these 3-braids together from left to right.}
        \label{fig:braidtwist}
\end{figure}
We now discuss lemmas, and theorems that are necessary for the proof of our first 3-braid model.
\begin{lem} (H. Cabrera-Ibarra \cite{cabrera2003classification}) \label{lem:uniquebraid}
For every 3-braid $B$ there exists a unique alternating diagram $AD=\mathcal{T}(a_1,...,a_n)$ with $a_i\cdot a_j \geq 0$ for all $i,j\in \{1,...,n \}$, and $a_i \neq 0$ for all $i>1$ and a unique integer, $k \in \mathbb{Z}$, such that $B = AD + kE$.
\end{lem}
A 3-braid is said to be in \textit{standard form} if it has the form described in Lemma \ref{lem:uniquebraid}. A \textit{flype move} is the transformation of a 3-braid $\mathcal{T}(a_1,...,a_r,a_{r+1},...,a_n)$ either into $\mathcal{T}(a_1,...,a_{r}+1,-1,1-a_{r+1},-a_{r+2}...,-a_n)-E$ or into  $\mathcal{T}(a_1,...,a_{r}-1,1,-1-a_{r+1},-a_{r+2},...,-a_n)+E$.
\begin{lem}\label{lem:flypedoesn'tchangebraid}
(H. Cabrera-Ibarra \cite{cabrera2003classification}) A flype move does not change the braid type. More precisely, if $n \in \mathbb{N}$ and $r \in \{1,...,n\}$, then 
\begin{align*}
\mathcal{T}(a_1,...,a_r,a_{r+1},...,a_n)&=\mathcal{T}(a_1,...,a_{r}+1,-1,1-a_{r+1},-a_{r+2}...,-a_n)-E\\
&=\mathcal{T}(a_1,...,a_{r}-1,1,-1-a_{r+1},-a_{r+2},...,-a_n)+E.
\end{align*}
\end{lem}
If $T$ is a rational tangle, then $N(T)$ is a well-studied families of links called \textit{four-plat links}. A key ingredient in the proof of our main theorem is the fact that a four-plat link can also be obtained as a certain ``closure'' of a 3-braid. Given a 3-braid $B$, the \textit{four-plat closure} $\mathcal{A}(B)$ of $B$ is a four-plat link obtained by joining endpoints of $B$ as
shown in Fig. \ref{fig:2bridgeclosure}. 

\begin{figure}[!ht]
  \centering
    \includegraphics[width=0.5\textwidth]{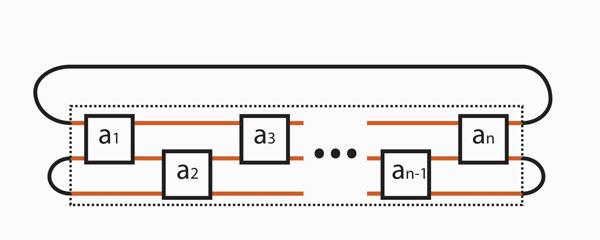}
        \caption{Performing the four-plat closure operation on a 3-braid gives a four-plat link.}
        \label{fig:2bridgeclosure}
\end{figure}
\begin{lem}\label{lem:propertyof2briddgeclosure}
(H. Cabrera-Ibarra \cite{cabrera2010braid}) The four-plat closure of a braid satisfies:

1. For every braid $B$ and every integer $k \in \mathbb{Z}, \mathcal{A}(B + 2kE) = \mathcal{A}(B).$ 

2. Given $a_2, ..., a_n \in \mathbb{Z}, \mathcal{A}(\mathcal{T} (0, a_2, ..., a_n)) = \mathcal{A}(\mathcal{T} (a_3, ..., a_n))$

\end{lem}

\begin{lem}\label{lem:unknotclosure}
If $k$ is an odd integer, then $\mathcal{A}(kE)$ is equivalent to the unknot $U$. If $k$ is an even integer, then $\mathcal{A}(kE)$ is equivalent to the
the unlink.
\end{lem}

Next, we restate two theorems about the four-plat closure of braids by H. Cabrera-Ibarra in terms of our sign conventions and notations\cite{ibarra2010algorithm}. The proofs are also reformulated and included for the convenience of the reader.
\begin{thm}(H. Cabrera-Ibarra \cite{ibarra2010algorithm})\label{thm:relatingbraidwithbridge}
Given a braid, $\mathcal{T}(a_1, ..., a_n) + kE,$ one has that:
\begin{center}
$\mathcal{A}(\mathcal{T} (a_1, ..., a_n) + kE)=
\begin{cases}
N((a_1 ,...,a_n )) & \text{if $k$ is even, $n$ is odd} \\
N((a_1 ,...,a_{n-1} )) & \text{if $k$ is even, $n$ is even} \\
N((-a_1 ,...,-a_{n-1})) & \text{if $k$ is odd, $n$ is odd} \\
N((-a_1 ,...,-a_n )) & \text{if $k$ is odd, $n$ is even} 
\end{cases}$
\end{center}

\end{thm}
\begin{proof}
Let $B = \mathcal{T} (a_1, ..., a_n) + kE$ be a 3-braid. The four-plat closure of $B$ can be viewed as the numerator closure of a tangle by labeling the NW, SW, NE and SE poles as shown in Fig. \ref{fig:viewedasnumerator}. \\
\textit{Case 1 ($k$ even):}
By Lemma \ref{lem:propertyof2briddgeclosure} we have that:
\begin{align*}
\mathcal{A}(\mathcal{T}(a_1, ..., a_n) + kE) = \mathcal{A}(\mathcal{T}(a_1, ..., a_n))
\end{align*}
If $n$ is odd, then labeling the points as shown in Fig. \ref{fig:viewedasnumerator} gives the numerator closure of the tangle $(a_1,...,a_n)$. Hence:
\begin{align*}
\mathcal{A}(\mathcal{T}(a_1,...,a_n) + kE) = N((a_1,...,a_n))
\end{align*}
If $n$ is even, then
\begin{align*}
\mathcal{A}(\mathcal{T}(a_1,...,a_n) + kE) = \mathcal{A}(\mathcal{T}(a_1,...,a_{n-1},0)) = N((a_1,...,a_{n-1}))
\end{align*}
\textit{Case 2 ($k$ odd)}: Suppose $k$ is odd, then $k = 2m+1$ for some $m \in \mathbb{Z}$ and by Lemma \ref{lem:propertyof2briddgeclosure} we have that:
\begin{align*}
\mathcal{A}(\mathcal{T} (a_1, ..., a_n) + (2m + 1)E) &= \mathcal{A}(\mathcal{T}(a_1, ..., a_n) + E + (2m)E)\\
&= \mathcal{A}(\mathcal{T}(a_1, ..., a_n) + E)
\end{align*}

If we label our points as shown in Fig. \ref{fig:viewedasnumerator}, then simplifying $E$ within the closure turns the numerator closure of a tangle into a denominator closure as shown by Fig. \ref{fig:simplifyingclosure}. The situation is analogous for $-E$, so we only prove the case with positive $E$. Now, if $n$ is even, then as shown in Fig. \ref{fig:simplifyingclosure}, it can be seen that we simply have the denominator closure of the tangle $(a_1,...,a_n,0$). Hence:
\begin{align*}
\mathcal{A}(\mathcal{T}(a_1, ..., a_n) + E) = D( (a_1, ... , a_n, 0)) = N((-a_1,...,-a_n)).
\end{align*}

If $n$ is odd, then
\begin{align*}
\mathcal{A}(\mathcal{T}(a_1, ..., a_n) + E) &= D( (a_1, ... , a_n))\\
&= D((a_1, ..., a_{n-1}, 0))\\
&= N((-a_1,...,-a_{n-1}))
\end{align*}

\end{proof}
\begin{figure}[!ht]
  \centering
    \includegraphics[width=0.5\textwidth]{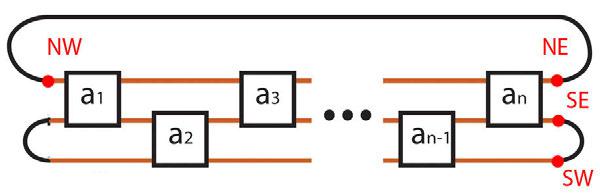}
        \caption{The four-plat closure of a braid viewed as the numerator closure of a tangle.}
        \label{fig:viewedasnumerator}
\end{figure}
\begin{figure}[!ht]
  \centering
    \includegraphics[width=0.7\textwidth]{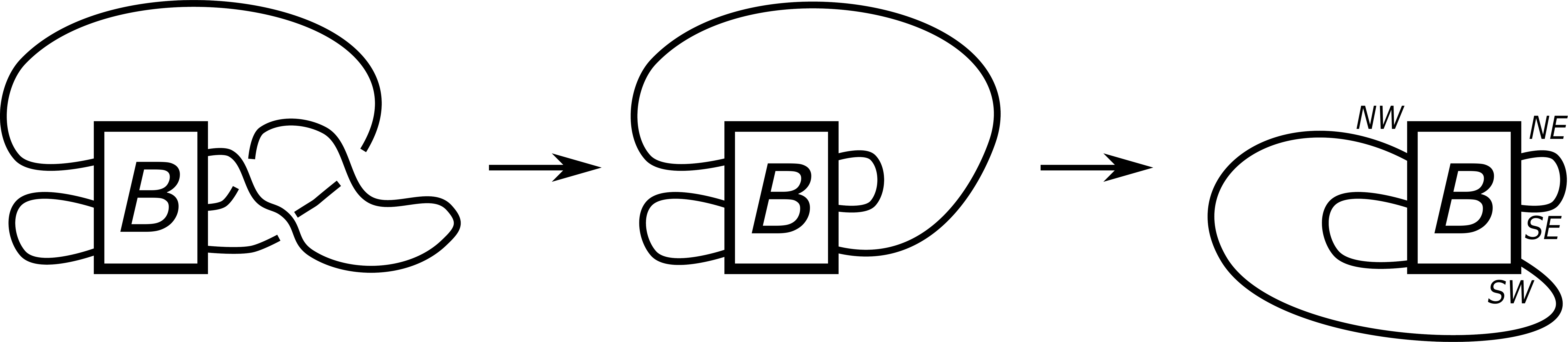}
        \caption{Simplifying the numerator closure into a denominator closure.}
        \label{fig:simplifyingclosure}
\end{figure}
\begin{thm}(H. Cabrera-Ibarra \cite{ibarra2010algorithm})\label{thm:braidclosureofknote}
Suppose that a braid $B$ admits a standard diagram with $a_1\neq 0$, then $\mathcal{A}(B)$ is the unknot if and only if $B$ admits one of the following diagrams:
\begin{align*}
&1. \mathcal{T}(\pm 1,a)+(2m)E \\
&2. \mathcal{T}(b)+(2m+1)E \\
&3. \mathcal{T}(\pm 1)+(2m)E
\end{align*}
where $a,b,m \in \mathbb{Z}.$
\end{thm}
\begin{proof}
We consider the 4 cases in Theorem \ref{thm:relatingbraidwithbridge} separately:\\

\textit{Case 1:}
Assume $k$ is even and $n$ is odd. Then,
\begin{align*}
\mathcal{A}(\mathcal{T} (a_1,...,a_n) + kE) = N((a_1,...,a_n))
\end{align*}

Since $a_{n-1}\neq 0$, it follows that $N((a_1,...,a_n))$ is the unknot if and only if $n = 1$ and $a_1 = \pm 1$. Therefore our original 3-braid is of the form $\mathcal{T}(\pm 1) + (2m)E$ where $m \in \mathbb{Z}$.\\

\textit{Case 2:}
Assume that $k$ and $n$ are both even. Then,
\begin{align*}
\mathcal{A}(\mathcal{T}(a_1,...,a_n) + kE) = N((a_1,...,a_{n-1}))
\end{align*}
We have that $N((a_1,...,a_{n-1}))$ is the unknot if and only if $n-1 = 1$ and $a_1 = \pm 1$. Therefore our original 3-braid is of the form $\mathcal{T}(\pm 1, a) + (2m)E$ where $a, m \in \mathbb{Z}$.\\

\textit{Case 3:}
Assume $k$ is odd and $n$ is odd. Then,
\begin{align*}
\mathcal{A}(\mathcal{T}(a_1,...,a_n) + kE) = N((-a_1,...,-a_{n-1}))
\end{align*}
We have two subcases. First, assume $n>1.$ We have that $N((-a_1,...,-a_{n-1}))$ is the unknot if and only if $n-1=1$. This leads to a contradiction since $n-1=1$ implies that $n=2$ and we assumed that $n$ was an odd number. Therefore we cannot have the unknot in this case.

If $n = 1$, then we see that $\mathcal{A}(\mathcal{T}(a_1) + kE)$ is the unknot. Hence, our braid is of the form $\mathcal{T}(b) + (2m + 1)E $ where $b, m \in \mathbb{Z}$.\\

\textit{Case 4:}
Assume $k$ is odd and $n$ is even. Then,
\begin{align*}
\mathcal{A}(\mathcal{T}(a_1,...,a_n) + kE) = N((-a_1,...,-a_n))
\end{align*}

We have that $N((-a_1,...,-a_n))$ is the unknot and $n = 1$. This is a contradiction since we assumed $n$ was an even number. Thus, we cannot get the unknot in this case.

\end{proof}

\subsection{Solutions to the 3-string tangle model}

We relax the assumptions made in Flapan et al.'s topological descriptions by allowing the possibility that the unknotted protein takes the shape of a four-plat closure of a 3-braid. The threading of a terminus through two loops can then be treated as a $\Gamma_2$ move as shown in Fig. \ref{fig:gammamoveonclosure}. We are interested in the knot types that result after performing the $\Gamma_2$ move in this fashion and we will show  that each of Flapan et al.'s configurations can be obtained as solutions to our tangle equation. We demonstrate this idea in Fig. \ref{fig:flapanconfig}. To be more precise, the tangle lying inside the dotted rectangle is a 3-braid. Outside of the dotted rectangle, the crossings involved in the $\Gamma_2$ move are highlighted in gray. After the $\Gamma_2$ move is performed, one can easily see that the knot we get is equivalent to one of Flapan-He-Wong's shapes after a slight deformation of the knot diagram. We now prove the main Theorem of this section.
\begin{figure}[!ht]
  \centering
    \includegraphics[width=0.7\textwidth]{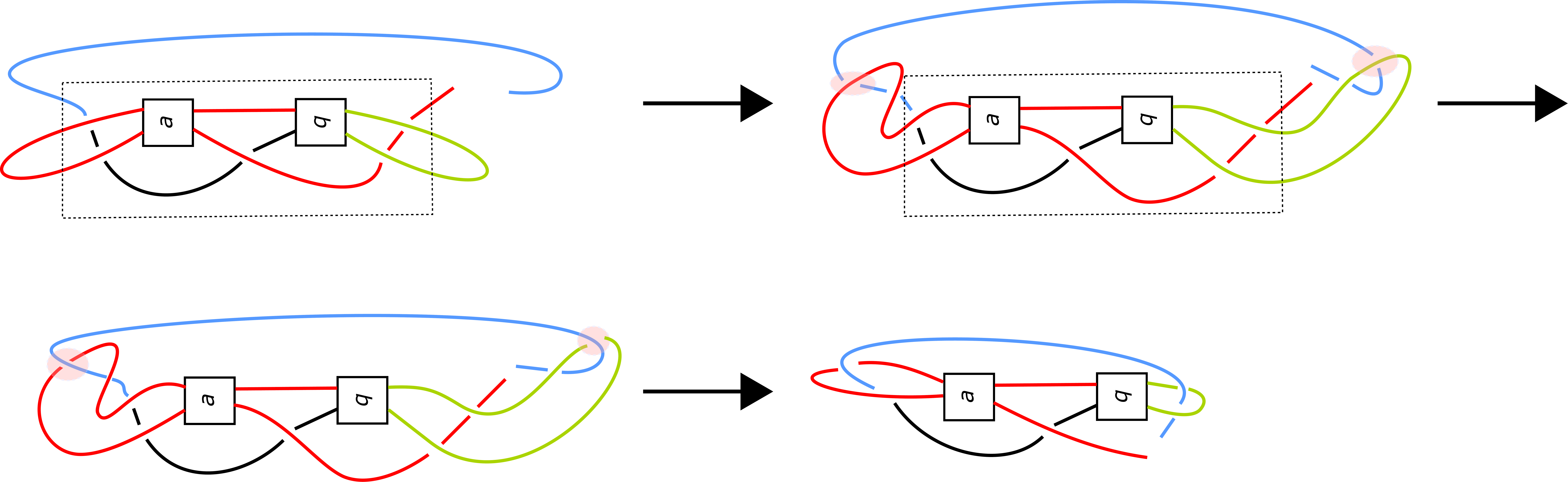}
        \caption{Obtaining Flapan-He-Wong's configurations from our 3-braid model}
        \label{fig:flapanconfig}
\end{figure}

\begin{thm}\label{thm:main}
Let $B$ be a $3$-braid such that $\mathcal{A}(B)$ is the unknot. Then, the knot that results from performing the $\Gamma_2$-move as shown in Fig. \ref{fig:2bridgeclosure} on the four-plat closure of $B$ is the unknot or one of the following forms $(n, m \in \mathbb{Z})$:
\begin{align*}
 &1.\;\ \text{Twist knots}  \\
 &2.\;\  N((-2,n+1,-3)) \;\ where \;\ n+1 \leq 0 \;\ \text{and their mirror images} \\
 &3.\;\   N((2,-n,-m,2)) \;\ where \;\ n,m \leq 0 \;\ \text{and their mirror images} \\
 &4.\;\   N((2,n-1,1,m,2)) \;\ \text{where} \;\ n,m > 0 \;\ \text{and their mirror images} 
\end{align*}
\end{thm}

\begin{proof}
By Lemma \ref{lem:uniquebraid}, we can put $B$ in standard form so that 
\begin{center}
    $B = \mathcal{T}(a_1, ..., a_n) + kE$ 
\end{center}
for some $k \in \mathbb{Z}$. By definition of the standard form of a braid, we have $a_i \cdot a_j > 0$ for all $i,j \in {1,...,n}$ and $a_i \neq 0$ for all $i > 1$. There are several cases to consider.

\textit{Case 1}: If $a_1\neq 0$ then by Theorem \ref{thm:braidclosureofknote}, there are three cases in which the four-plat closure gives the unknot: either $B = \mathcal{T}(b) + (2m + 1)E$, $B = \mathcal{T}(\pm 1,a) + (2m)E$, or $B = \mathcal{T} (\pm 1) + (2m)E.$ We will discuss what knots result from the $\Gamma_2$ move in each case separately.

\textit{Subcase 1:} Assume our braid is of the form $\mathcal{T}(b) + (2m + 1)E$. We start our calculations by noticing that the closure of an arbitrary braid, $\mathcal{T}(a_1, ..., a_n) + kE$, after the $\Gamma_2$ move is just the closure of a new braid, $\mathcal{T}(-2) + \mathcal{T} (a_1, ..., a_n) + kE + \mathcal{T} (2)$ as shown in Fig. \ref{fig:gammamoveonclosure}.
Thus, if we start out with $\mathcal{T}(b) + (2m + 1)E$, then we have the four-plat closure of the braid below to consider after the $\Gamma_2$ move is performed:
\begin{equation}\label{eq:1}
\mathcal{T}(-2) + \mathcal{T}(b) + (2m + 1)E+ \mathcal{T} (2) = \mathcal{T} (b - 2) + (2m + 1)E + \mathcal{T} (2).
\end{equation}
Applying $2m + 1$ flype moves to equation \ref{eq:1} in between $(2m + 1)E$ and $\mathcal{T}(2)$ and performing Reidemeister moves gives:
\begin{align*}
\mathcal{T}(b -2) + (2m + 1)E + \mathcal{T}(2) &=\mathcal{T}(b - 2) + \mathcal{T}(0, -2) + (2m + 1)E. \\ &= \mathcal{T}(b -2, -2) + (2m + 1)E
\end{align*}
Now that we have simplified the new braid, we consider its four-plat closure to determine what knot results. By Theorem \ref{thm:relatingbraidwithbridge} we have that:
\begin{align*}
\mathcal{A}(\mathcal{T}(b - 2, -2) + (2m + 1)E) &= N ((2 - b, 2)).
\end{align*}
The numerator closure of $(2 - b, 2)$ can easily be visualized and results in twist knots except when $b = 2$ or 3, in which case we get the unknot.

\textit{Subcase 2:} Suppose that our braid is of the form $\mathcal{T} (\pm 1, a_2) + (2m)E$. Then, we have the four-plat closure of the braid below to consider after the $\Gamma_2$ move is performed:

\begin{align*}
\mathcal{T}(-2) + \mathcal{T}(\pm 1, a_2) + (2m)E + \mathcal{T}(2) = \mathcal{T}(\pm1 - 2, a_2) + (2m)E + \mathcal{T}(2).
\end{align*}
Applying $2m$ flype moves in between $(2m)E$ and $\mathcal{T}(2)$ gives:
\begin{align*}
\mathcal{T}(\pm 1 - 2, a_2) + (2m)E + \mathcal{T}(2) = \mathcal{T}(\pm1 - 2, a_2, 2) + (2m)E
\end{align*}
By Theorem \ref{thm:relatingbraidwithbridge} we have that:
\begin{align*}
\mathcal{A}(\mathcal{T}(\pm1 - 2, a_2, 2) + (2m)E) = N ((\pm1 - 2, a_2, 2))
\end{align*}

Now, for the new first entry, $a_1'$, in our new braid we have that either $a_1' = 1-2 = -1$ or $a_1' = -1-2=-3$.
If $a_1' =-1$, then $a_2 > 0$ because our braid was assumed to be in standard form and we have:
\begin{align*}
N((-1,a_2,2)) &= N\left(2+\dfrac{1}{a_2+\dfrac{1}{-1}}\right)\\
&= N\left(2+\dfrac{1}{a_2-1}\right).
\end{align*}
In the tangle above, one $-1$ horizontal twist followed by $a_2$ vertical twists is the same as $a_2 -1$ vertical twists. Thus, this can be viewed as the numerator closure of $(a_2 -1, 2)$, which can be easily visualized. It is the unknot when $a_2 = 1$,  or a twist knot when $a_2 > 1$.

Now, if $a'_1 = -1-2 = -3$, then $a_2 < 0$ because our braid was assumed to be in standard form. Now, it is desirable to give the solutions where the diagrams have minimal number of crossings. We make use of
the fact that $N((a_1,...,a_n)) = N((a_n,...,a_1))$, so we have that $N((-3,a_2,2)) =
N((2,a_2,-3))$. This gives the following:
\begin{align}
N((2,a_2,-3))&=N\left(-3+\dfrac{1}{a_2+\dfrac{1}{2}}\right)\\
&= N\left(-3+\dfrac{1}{(a_2+1)+\dfrac{1}{-2}}\right)\label{eq:e3}.
\end{align}
Equation \ref{eq:e3} implies that $N((2,a_2,-3))=N((-2,a_2 +1,-3))$. We also have that $a_2 + 1 \leq 0$. Therefore, if $a_2 = -1,$ then $N((-2,a_2 + 1,-3)) = N((-2,0,-3)) = N((-5))$ which is the mirror image of the torus knot $5_1$. Otherwise, if $a_2 < -1$, then $a_2 + 1 \leq  -1$. Furthermore, the numerator closure $N((-2,a_2 + 1,-3))$ is in standard form and gives a minimal diagram.

\textit{Subcase 3}:
Now, assume our braid is of the form $\mathcal{T}(\pm 1)+(2m)E$. Then, we have the four-plat closure of the braid below to consider after the $\Gamma_2$-move is performed:
\begin{align*}
\mathcal{T}(-2) + \mathcal{T}(\pm 1) + (2m)E + \mathcal{T}(2) = \mathcal{T}(\pm1 - 2) + (2m)E + \mathcal{T}(2)
\end{align*}
Applying $2m$ flype moves in between $(2m)E$ and $\mathcal{T}(2)$ gives:
\begin{align*}
\mathcal{T}(\pm1 - 2) + (2m)E + \mathcal{T}(2) &= \mathcal{T}(\pm1 - 2 + 2) + (2m)E \\ &= \mathcal{T}(\pm1) + (2m)E
\end{align*}
By Theorem \ref{thm:braidclosureofknote} we have that $\mathcal{A}(\mathcal{T}(\pm1) + (2m)E)$ is the unknot. So, in this subcase we can only get the unknot. 

This concludes the case when $a_1\neq 0$ and we now consider the case when $a_1 = 0$.

\textit{Case 2:}
Assume $a_1 = 0$ in the original braid, $\mathcal{T}(a_1, ..., a_n) + kE.$ 

\textit{Subcase 1:}
Assume $n = 2$, then the braid is of the form $\mathcal{T}(0, a_2) + kE$, and by Lemma \ref{lem:propertyof2briddgeclosure}, it follows that $\mathcal{A}(\mathcal{T} (0, a_2) + kE) = \mathcal{A}(kE)$. 

By Lemma \ref{lem:unknotclosure}, if $k$ is even then the closure above is the unlink. If $k$ is odd, then the closure is the unknot. Our model is only concerned with cases where we start out with the unknot so we only consider the case where $k$ is odd.

Assume $k$ is odd. Then, $k = 2m+1$ for some $m \in \mathbb{Z}$. With our new closure shown in Fig. \ref{fig:gammamoveonclosure} we consider the following braid:
\begin{align*}
\mathcal{T}(-2) + \mathcal{T}(0, a_2) + (2m + 1)E + \mathcal{T}(2)\\ = \mathcal{T}(-2, a_2) + \mathcal{T}(0, -2) + (2m + 1)E\\
=\mathcal{T}(-2,a_2-2)+(2m+1)E
\end{align*}
By Theorem \ref{thm:relatingbraidwithbridge} we have that:

\begin{align*}
\mathcal{A}(\mathcal{T}(-2,a_2-2)+(2m+1)E)&=N((2,2-a_2))\\
&=N\left(2-a_2+\dfrac{1}{2}\right).
\end{align*}
The numerator closure of $(2, 2 - a_2)$ can be easily visualized and is a twist knot except in the cases where $a_2 = 2$ or $a_2 = 3$ which gives the unknot. This finishes the case where $n = 2.$

\textit{Subcase 2:}
Now, we assume $n > 2$, then the braid is of the form $\mathcal{T}(0, a_2, a_3, ..., a_n) + kE$. By Lemma \ref{lem:propertyof2briddgeclosure} it follows that the four-plat closure of our braid is $\mathcal{A}(\mathcal{T} (0, a_2, a_3, ..., a_n) + kE) = \mathcal{A}(\mathcal{T}(a_3, ..., a_n) + kE)$. Since our braid was in standard form, we have that $a_3\neq 0$. Therefore, by Theorem \ref{thm:braidclosureofknote}, the four-plat closure is the unknot only when the braid has one of the following forms:

\begin{align*}
\mathcal{T}(0, a_2, b) + (2m + 1)E \\
\mathcal{T}(0, a_2, \pm 1, a_4) + (2m)E \\
\mathcal{T}(0, a_2, \pm1) + (2m)E.
\end{align*}
where $a_2, a_4, b, m \in \mathbb{Z}$. We consider each case separately.\\

\textit{Subcase 2A}:
Assume our braid is of the form $\mathcal{T} (0,a_2,b)+(2m+1)E$, then our new closure pictured in Fig. \ref{fig:gammamoveonclosure} allows us to consider the four-plat closure of the new braid:

\begin{align*}
\mathcal{T}(-2) + \mathcal{T}(0, a_2, b) + (2m + 1)E + \mathcal{T}(2) \\ = \mathcal{T}(-2, a_2, b) + \mathcal{T} (0, -2) + (2m + 1)E\\
= \mathcal{T}(-2, a_2, b, -2) + (2m + 1)E
\end{align*}

By Theorem \ref{thm:relatingbraidwithbridge}, it follows that $\mathcal{A}(\mathcal{T}(-2, a_2 , b, -2) + (2m + 1)E ) = N ((2, -a_2 , -b, 2))$.

If $b < 0$, then $a_2 < 0$ because the braid we started out with was in standard form. In this case our numerator closure is in standard form and we have a minimal diagram. Now, if $b > 0$, then $a_2 > 0$ because the braid we started our with was in standard form. In this case, our numerator closure is not in standard form. To put the numerator closure in standard form we observe that:
\begin{align*}
N((2,-a_2,-b,2)) &= N\left(2+\dfrac{1}{-b+\frac{1}{-a_2+\frac{1}{2}}}\right)\\
&= N\left(2+\dfrac{1}{-b+\frac{1}{(-a_2+1)+\frac{1}{-2}}}\right).
\end{align*}
We can see that $N((2,-a_2,-b,2)) = N((-2,-a_2+1,-b,2))$. We also have that $N((-2,-a_2+1,-b,2)) = N((2,-b,-a_2+1,-2))$ giving:
\begin{align*}
N((2,-b,-a_2+1,-2)) &= N\left(2+\dfrac{1}{(-a_2+1)+\frac{1}{-b+\frac{1}{2}}}\right)\\
&= N\left(-2+\dfrac{1}{(-a_2+1)+\frac{1}{(-b+1)+\frac{1}{-2}}}\right).
\end{align*}
Hence, $N((2,-b,-a_2 +1,-2)) = N((-2,-b+1,-a_2 +1,-2))$ and our numerator closure is in standard form in this case also since $-b + 1, -a_2 + 1 \leq 0.$ This concludes Subcase 2A.

\textit{Subcase 2B:}
Assume our braid is of the form $\mathcal{T}(0, a_2, \pm1, a_4) + (2m)E$. It follows that

\begin{align*}
\mathcal{T}(-2) + \mathcal{T}(0, a_2, \pm1, a_4) + (2m)E + \mathcal{T}(2) \\ = \mathcal{T}(-2, a_2, \pm 1, a_4) + \mathcal{T}(2) + (2m)E\\
= \mathcal{T}(-2, a_2, \pm1, a_4, 2) + (2m)E.
\end{align*}
By Theorem \ref{thm:relatingbraidwithbridge}, we have that 
\begin{center}
    $\mathcal{A}(\mathcal{T} (-2, a_2 , \pm1, a_4 , 2) + (2m)E ) = N ((-2, a_2 , \pm1, a_4 , 2))$,
\end{center}
and it follows that either $a_3 = 1$ or $a_3 = -1$. If $a_3 =1$, then $a_2,a_4 >0$ and we have:

\begin{align}
N((-2,a_2,1,a_4,2)) &= N\left(2+\dfrac{1}{a_4+\frac{1}{1+\frac{1}{a_2+\frac{1}{-2}}}}\right)\\
&= N\left(2+\dfrac{1}{a_4+\frac{1}{1+\frac{1}{(a_2-1)+\frac{1}{2}}}}\right) \label{e6}.
\end{align}
From equation \ref{e6}, we see that $N((-2,a_2,1,a_4,2)) = N((2,a_2-1,1,a_4,2))$.

Since $a_4 > 0$ and $a_2 - 1 \geq 0$, it follows that the numerator closure $N((2,a_2 - 1, 1, a_4, 2))$ is in standard form and therefore a minimal diagram.

Finally, if $a_3 = -1$, then $a_2, a_4 < 0$ and we have that $N((-2,a_2,-1,a_4,2)) = N((2,a_4,-1,a_2,-2))$. Expanding this out we get:
\begin{align}
N((2,a_4,-1,a_2,-2)) &= N\left(-2+\dfrac{1}{a_2+\frac{1}{-1+\frac{1}{a_4+\frac{1}{2}}}}\right)\\
&= N\left(-2+\dfrac{1}{a_2+\frac{1}{-1+\frac{1}{(a_4+1)+\frac{1}{-2}}}}\right) \label{e10}.
\end{align}
Equation \ref{e10} implies that $N((2,a_4,-1,a_2,-2))=N((-2,a_4 +1,-1,a_2,-2))$. Since $a_4 + 1 \leq 0$ and $a_2 < 0$, it follows that the numerator closure $N((-2,a_4 + 1, -1, a_2, -2))$ is in standard form.

\textit{Subcase 2C}:
Now suppose our braid is of the form $\mathcal{T}(0, a_2, \pm1) + (2m)E$. It follows that

\begin{align*}
\mathcal{T}(-2) + \mathcal{T}(0, a_2, \pm1) + (2m)E + \mathcal{T}(2) \\= \mathcal{T}(-2, a_2, \pm1) + \mathcal{T}(2) + (2m)E\\
= \mathcal{T}(-2, a_2, \pm1 + 2) + (2m)E.
\end{align*}
It follows that $\mathcal{A}(\mathcal{T} (-2, a_2 , \pm 1 + 2) + (2m)E ) = N ((-2, a_2 , \pm1 + 2))$.

Now, for the third entry, $a_3'$, in our new braid we have that either $a_3' = 1 + 2 = 3$ or
$a_3' = -1 + 2 = 1.$
If $a_3' =1+2=3$, then $a_2 >0$ and we have:
\begin{align}
N((-2,a_2,3))&=N\left(3+\dfrac{1}{a_2+\frac{1}{-2}}\right)\\
&=N\left(3+\dfrac{1}{(a_2-1)+\frac{1}{2}}\right)\label{e19}.
\end{align}

If we look at equation \ref{e19} we can see that $N((-2,a_2,3)) = N((2,a_2 - 1,3)).$ We also have that $a_2 - 1 \geq 0.$ Therefore, if $a_2 = 1$, then $N((2,a_2 - 1,3)) = N((2,0,3)) = N((5))$ which is the torus knot $5_1$. Otherwise, if $a_2 > 1$, then $a_2 - 1 \geq 1$ and the numerator closure $N ((2, a_2 - 1, 3))$ is in standard form and gives a minimal diagram.

If $a_3' =-1+2=1$, then $a_2 <0$ and we note that $N((-2,a_2,1)) = N((1,a_2,-2))$. It follows that

\begin{align*}
N((1,a_2,-2))&=N\left(-2+\dfrac{1}{a_2+\frac{1}{1}}\right).
\end{align*}
In the tangle $N((1,a_2,-2))$, we note that one +1 horizontal twist followed by one $a_2$ vertical twists is the same as $a_2 + 1$ vertical twists. Thus, this can be viewed as the numerator closure of $(a_2 + 1, -2)$ which can be easily visualized. It is the unknot when $a_2 =-1$ or $a_2 =-2$. If $a_2<-2$ then we get a twist knot.

This completes the proof of Theorem \ref{thm:main}.
\end{proof}
\begin{figure}[!ht]
  \centering
    \includegraphics[width=0.5\textwidth]{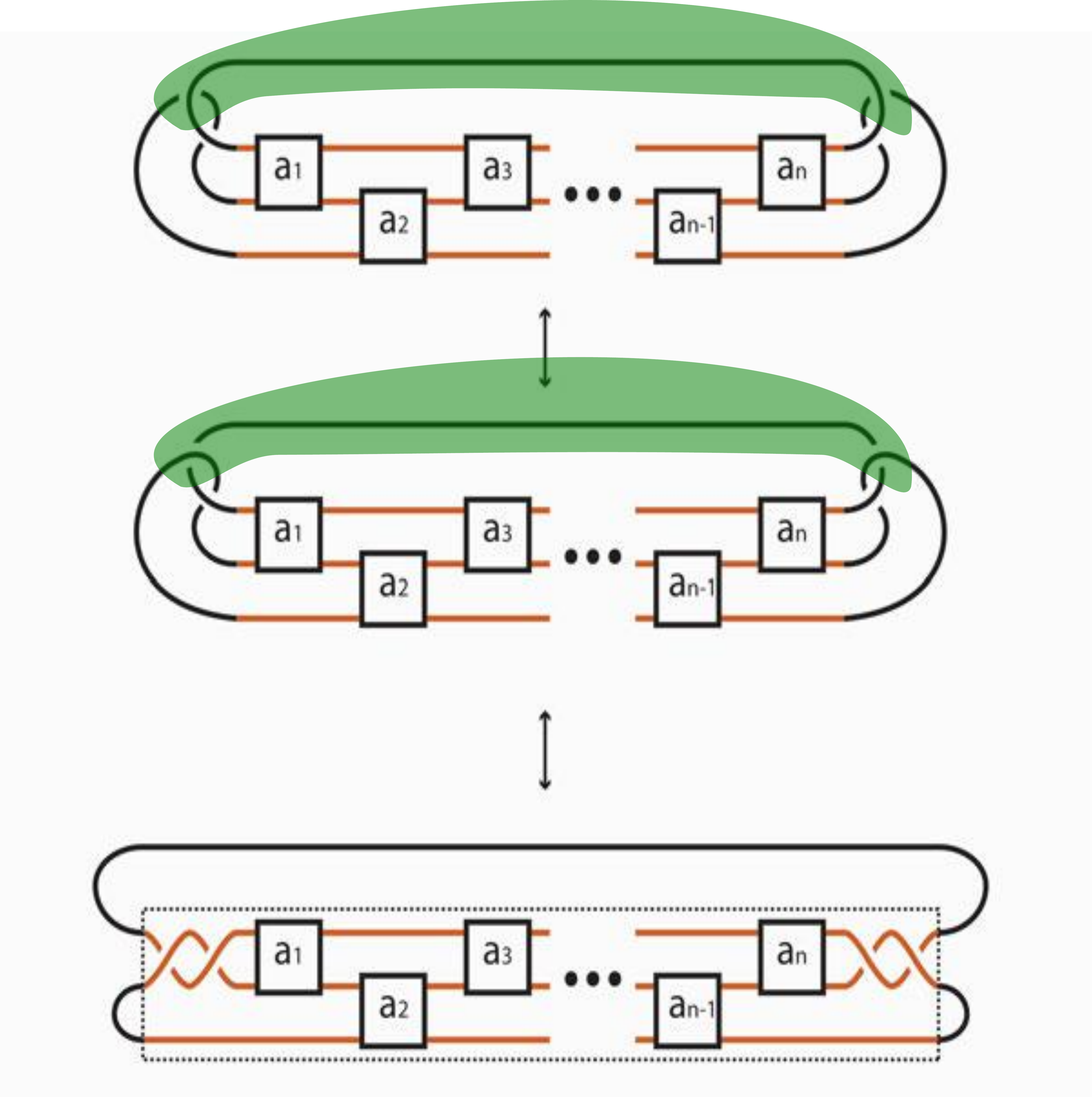}
        \caption{The four-plat closure of a braid transformed by a $\Gamma_2$-move into the closure of a new braid.}
        \label{fig:gammamoveonclosure}
\end{figure}
\begin{figure}[!ht]
  \centering
    \includegraphics[width=0.3\textwidth]{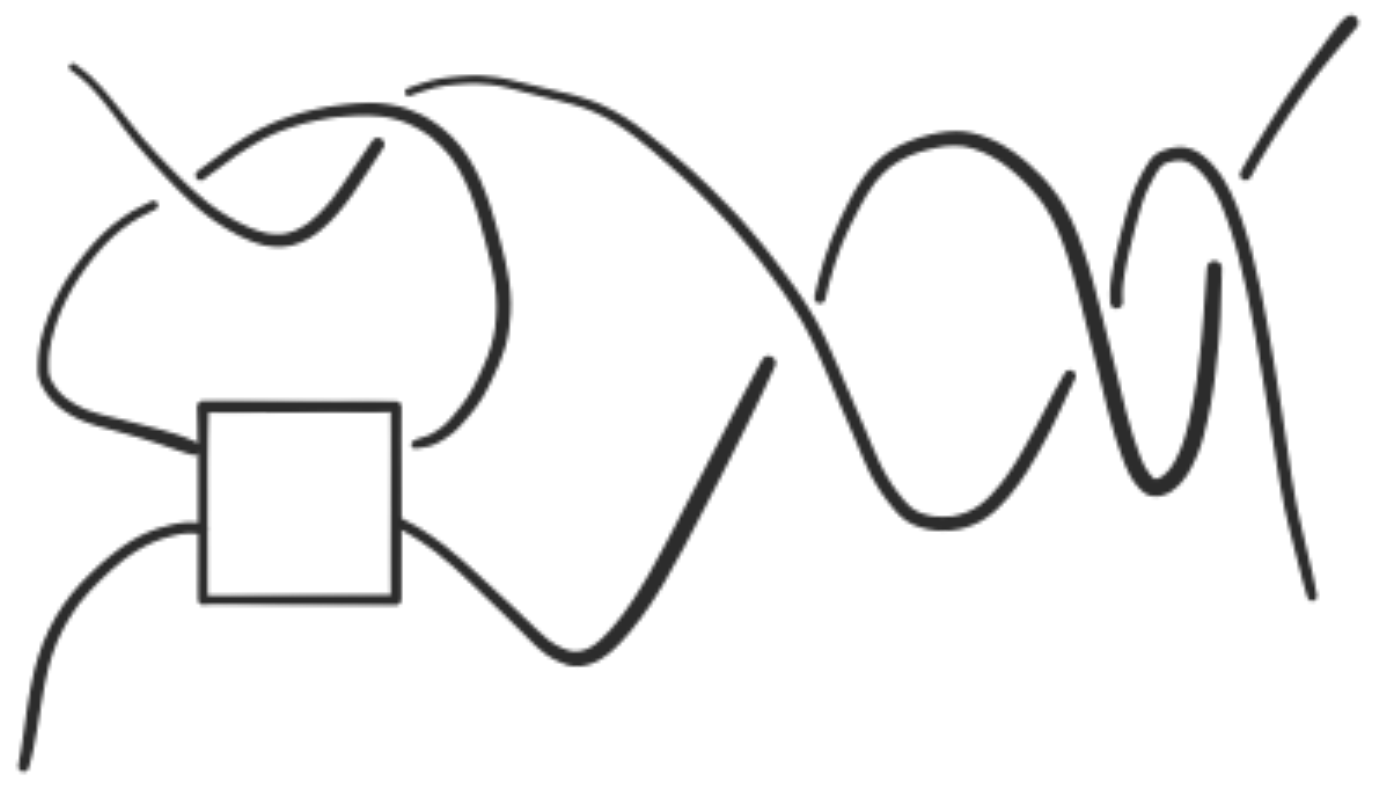}
        \caption{The 3-braid tangle solution $N((-2,n+1,3))$, which is the second family found in Theorem \ref{thm:main}. Here, the box represents the integer tangle $n+1$.}
        \label{fig:secondsol}
\end{figure}
\begin{figure}[!ht]
  \centering
    \includegraphics[width=0.3\textwidth]{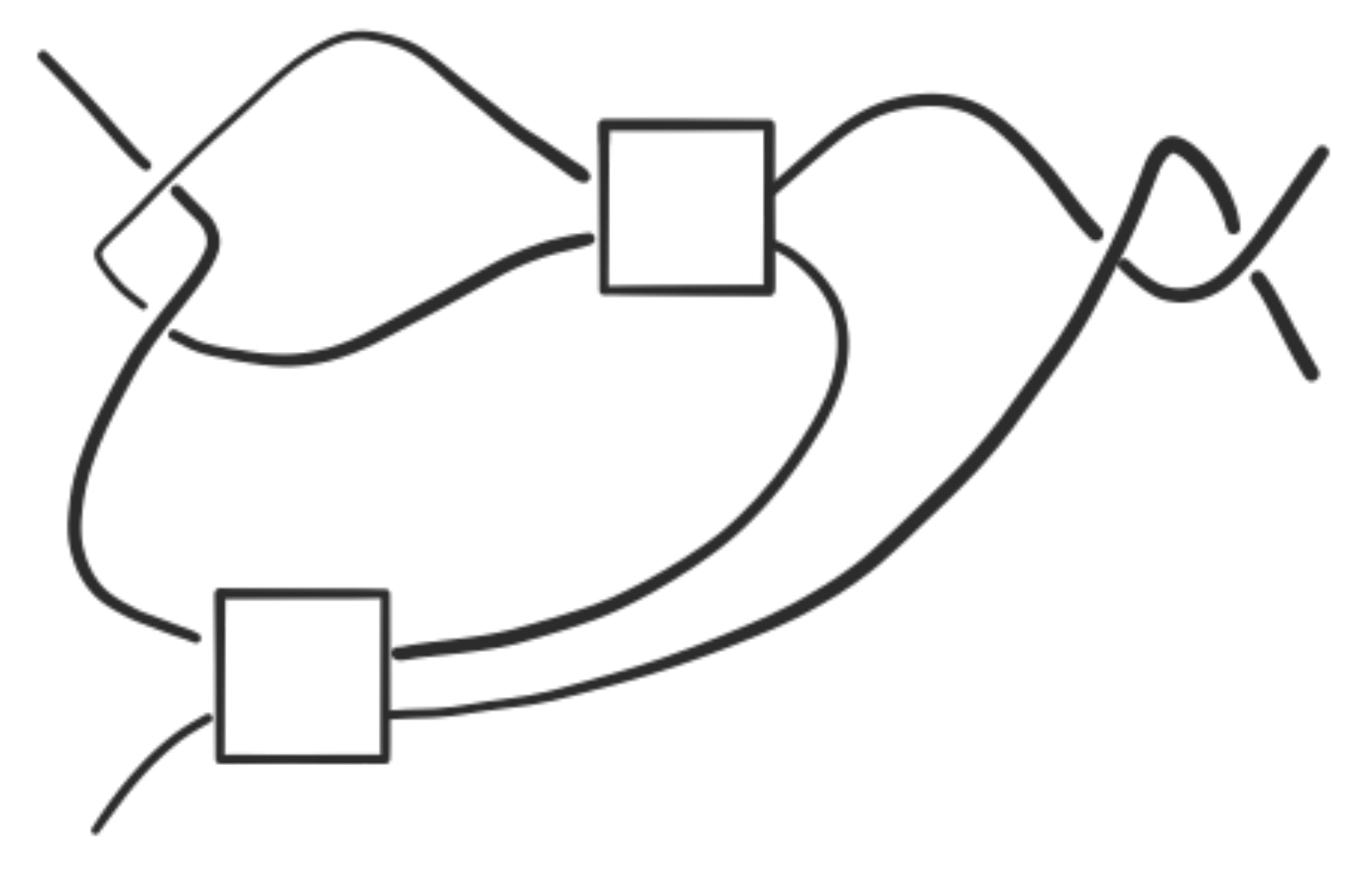}
        \caption{The 3-braid tangle solution $N((2,-n,-m,2))$, which is the third family found in Theorem \ref{thm:main}. The top (resp. bottom) box represents the integer tangle $-n$ (resp. $-m$).}
        \label{fig:thirdsol}
\end{figure}
\begin{figure}[!ht]
  \centering
    \includegraphics[width=0.3\textwidth]{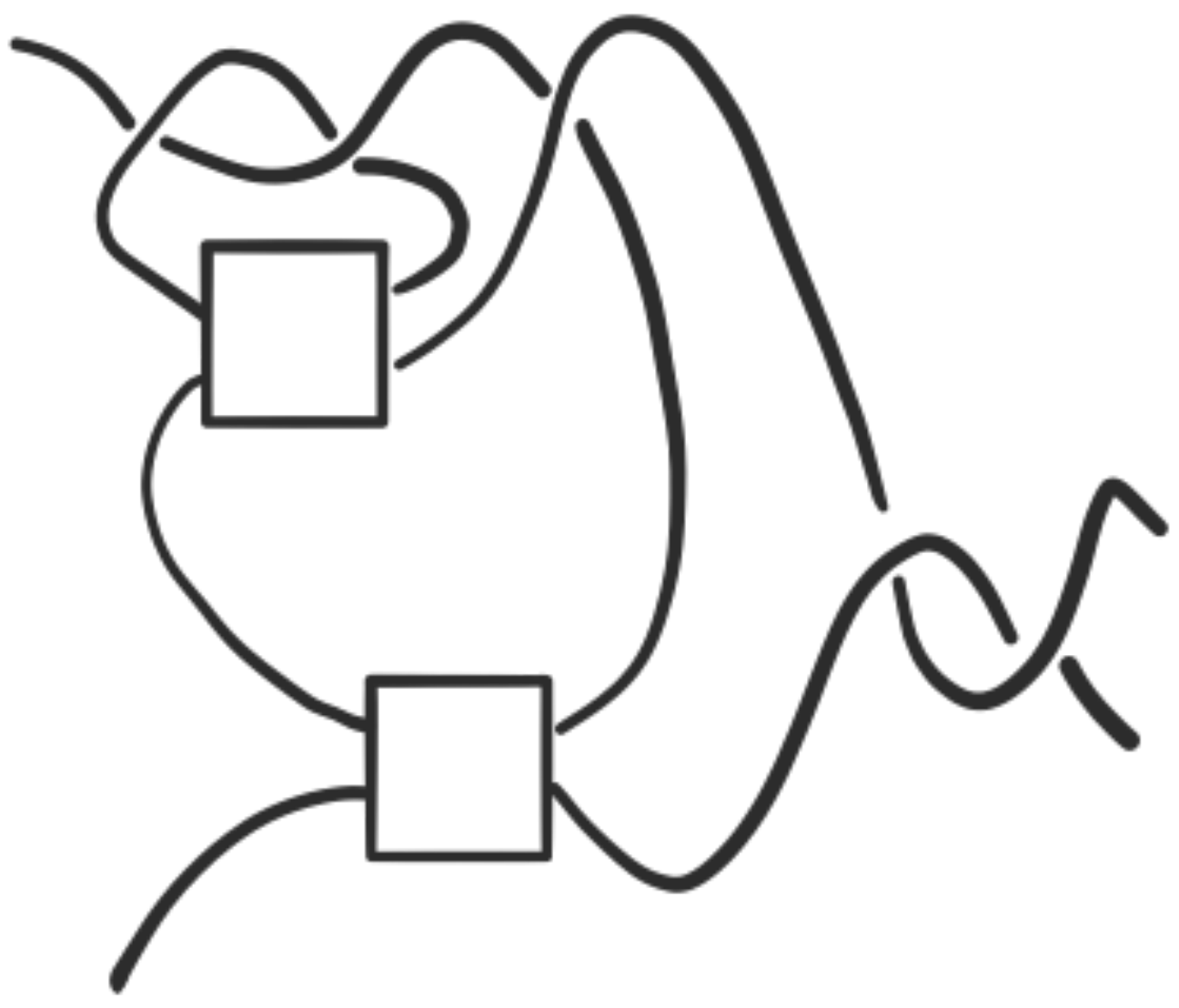}
        \caption{The 3-braid tangle solution $N((2,n-1,1,m,2))$, which is the fourth family found in Theorem \ref{thm:main}. The top (resp. bottom) box represents the integer tangle $-n-1$ (resp. $m$).}
        \label{fig:finsol}
\end{figure}

Twist knots were discussed before as they arose in our 2-string tangle model. The remaining three families of solutions consist of numerator closures of the tangles depicted in Fig. \ref{fig:secondsol}, \ref{fig:thirdsol}, and \ref{fig:finsol}. We can obtain the knot $+5_1$ and its mirror image, which have not been found in proteins from the second family in Theorem \ref{thm:main} by setting $n+1 = 0$. The only knots with 6 crossings are $+6_1,-6_1,+6_2, -6_2,$ and $6_3$. The $+6_1$ and $-6_1$ knots are twist knots. The $+6_2$ and $-6_2$ knots are particular instances of the second family in Theorem \ref{thm:main}. The knot $6_3$ belongs to the third family in Theorem \ref{thm:main}. The first interesting knot in the fourth family from Theorem \ref{thm:main} that has not been mentioned before is the knot $+7_7$. Paying attention to the orientation, Theorem \ref{thm:main} also gives 4-plat knots $K$ such that $d_{\overrightarrow{\Gamma_2}}(K,unknot) =1.$ This oriented version of distance is discussed in the following section.

\section{Oriented moves arising in proteins}\label{section:orientedmoves}
\subsection{The oriented $\Gamma_2$ move}
If one assigns a direction to the folded DehI protein, one sees that the green region has an arrow pointed in the opposite direction compared to the arrow on red region. One can then think of the $6_1$-knot on DehI as being obtained from the unknotted protein by letting the blue strand passes through the red and green portion, where the arrows on are forced to point in certain fixed directions. Such a local move on knot diagrams, which we will denote as the $\overrightarrow{\Gamma}_2$-move, has been studied by Shibuya and Kanenobu \cite{shibuya2000local,kanenobu2014sh}. Before we state our result, we remind the readers of the definitions of relevant local moves. A \textit{coherent band surgery} is a local move depicted on the left of Figure \ref{fig:bandanddelta}. A \textit{$\Delta$-move} is a local move depicted on the right of Figure \ref{fig:bandanddelta}. We denote by $d_{band}(K,K')$ (resp. $d_{\Delta}(K.K')$) the minimum number of coherent band surgeries needed (resp. the minimum number of $\Delta$ moves needed) to transform $K$ into $K'$. Reidemeister moves are allowed in the intermediate stages. We use the following observations to estimate $d_{\overrightarrow{\Gamma}_2}(K,K')$:
\begin{figure}[ht!]
  \centering
    \includegraphics[width=0.5\textwidth]{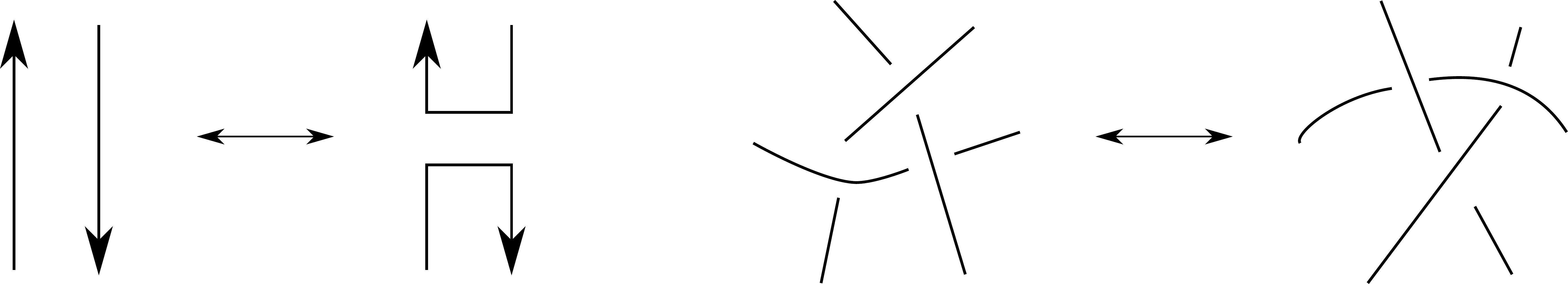}
      \caption{(Left) A coherent band surgery. (Right) A delta move. .}\label{fig:bandanddelta}
\end{figure}
\begin{prop}[Prop. 4.1 and Prop. 4.3 from \cite{kanenobu2014sh}]
Let $K$ and $K'$ be two knots or two knotoids. Then, the following statements hold:

\begin{enumerate}
    \item $2d_{band}(K.K')\leq d_{\overrightarrow{\Gamma}_2}(K,K') \leq d_{\Delta}(K,K')$.
    \item $d_{\overrightarrow{\Gamma}_2}(K,K')\leq d(K,K') \leq 2d_{\overrightarrow{\Gamma}_2}(K,K').$
\end{enumerate}

\end{prop}

The main result of this section is the table at the end of the section displaying the values $d_{\overrightarrow{\Gamma}_2}(K,K')$ for various knots $K,K'$. If the entries are blue, then the distances are found by drawing explicit diagrams and finding where $\overrightarrow{\Gamma}_2$ moves can be preformed (see Fig. \ref{fig:41to31} and Fig. \ref{fig:41to63}). The justification for the values colored red is the triangle inequality (as these distances satisfy axioms of metric spaces).

\begin{figure}[ht!]
  \centering
    \includegraphics[width=0.5\textwidth]{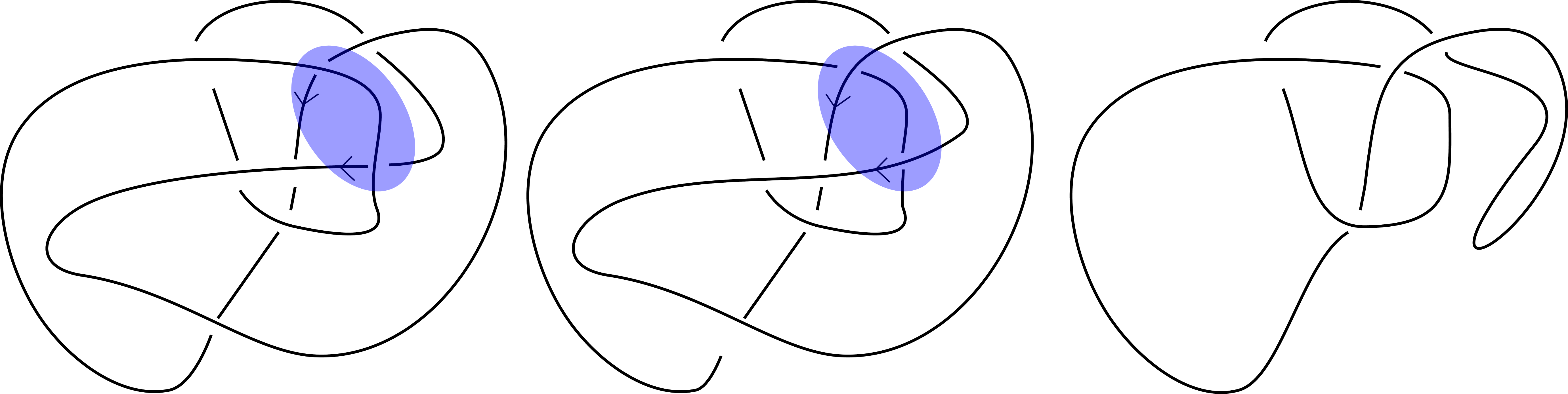}
      \caption{(Left) A non-minimal diagram $D$ of the figure-eight knot. (Middle) A diagram $D'$ after an oriented $\Gamma_2$-move move has been performed on $D$. (Right) After some Reidemeister moves, we get a trefoil knot. The trefoil of opposite handedness can be obtained this way if we switch all the crossings of $D$ and repeat the process.}\label{fig:41to31}
\end{figure}
\begin{figure}[ht!]
  \centering
    \includegraphics[width=0.5\textwidth]{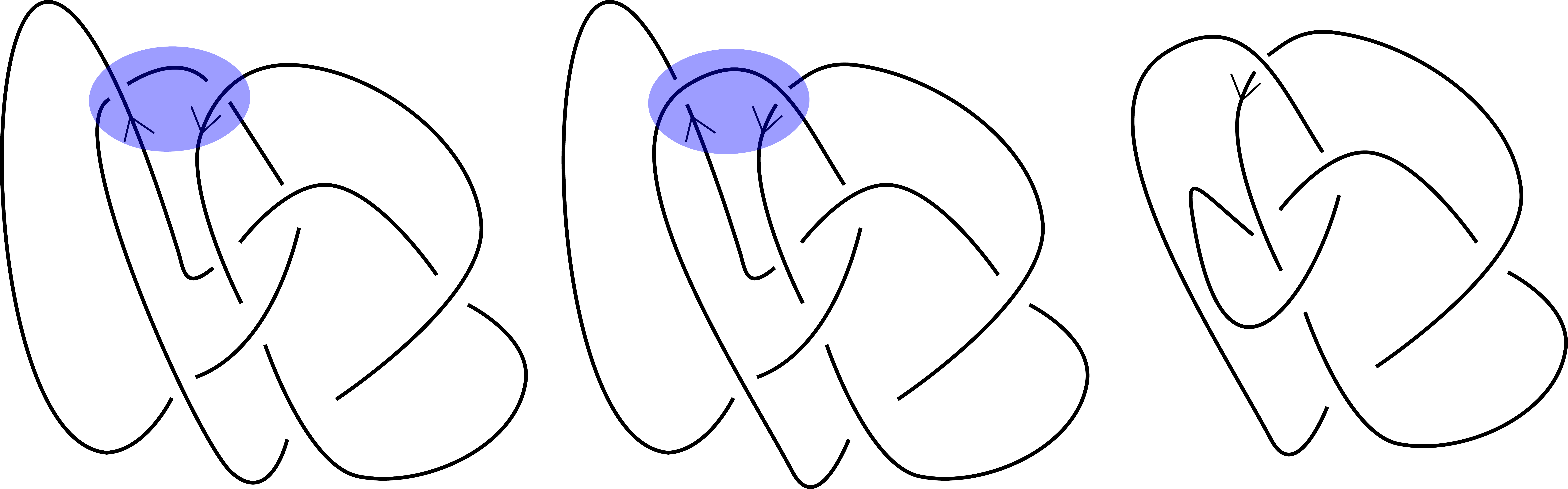}
      \caption{(Left) A non-minimal diagram $D$ of the figure-eight knot. (Middle) A diagram $D'$ after an oriented $\Gamma_2$-move move has been performed on $D$. (Right) After some Reidemeister moves, we get the $6_3$ knot.}\label{fig:41to63}
\end{figure}

\begin{center}
    \begin{tabular}{| l | l | l | l | l | l | l | l | }
    \hline
  & $3_1$ & $4_1$ & $5_1$ & $5_2$ & $6_1$ & $6_2$ & $6_3$  \\ 
  \hline
       $0_1$ & 1 & 1& 2&1  &1  &1 & 1\\
  \hline
   $3_1$ & 0  & \textcolor{blue}{1}  & 1 & 1 & 2 & 1 & 1  \\ \hline
$3_1^*$& 2 & \textcolor{blue}{1} & 3 & 2 & 1-2 & 2 & 1 \\ \hline
   $4_1$ & \textcolor{blue}{1} & 0 & \textcolor{red}{2}& 1-2 & 1 & 1& \textcolor{blue}{1}\\
    \hline
       $5_1$ & 1 & \textcolor{red}{2} & 0& 1 & 2-3 &1-2 &2 \\
    \hline  
    $5_1^*$ & 3 & 2-3 & 4& 3 & 2-3 & 3 & 2\\
    \hline
       $5_2$ & 1 &1-2 &1 & 0 & 1-2 &1-2 &1-2  \\
    \hline
       $5_2^*$ & 2 &1-2 &3 & 2 & 1-2 &2 &1-2  \\
    \hline
       $6_1$ & 2 &1 & 2-3& 1-2 & 0 &1 & 1-2 \\
    \hline
       $6_1^*$ & 1-2 &1 & 2-3& 1-2 & 1 &1-2 &1-2  \\
    \hline
       $6_2$ & 1 &1 & 1-2&1-2  & 1 & 0& 1-2\\
    \hline
       $6_2^*$ & 2 & 1& 3& 2 & 1-2 & 2& 1-2\\
    \hline
    \end{tabular}
\end{center}
\subsection{Signed forbidden moves}
In \cite{barbensi2021f}, the authors considered the forbidden distances between two knotoids. To investigate chirality bias, we can analyze the signed versions of the forbidden move.

The authors of \cite{barbensi2021f} pointed out that a forbidden moved performed on a knotoid corresponds to a crossing change performed on the knots obtain by closing up the endpoints. We can use the same idea to bound the signed forbidden moves. The first author's thesis \cite{darcythesis} will be useful in calculating this quantity.

\begin{exmp}
Any positive forbidden move performed on the knotoid $2_1$ induces a crossing change of its overpass closure $+3_1$. However, as mentioned in the first author's thesis \cite{darcythesis}, it is not possible to transform $+3_1$ by positive crossing changes only. Thus, even though the forbidden distance between $2_1$ and the trivial knotoid is 1, it is not possible to go from $2_1$ to the trivial knotoid by only using positive forbidden moves.
\end{exmp}

\begin{figure}[ht!]
  \centering
    \includegraphics[width=0.5\textwidth]{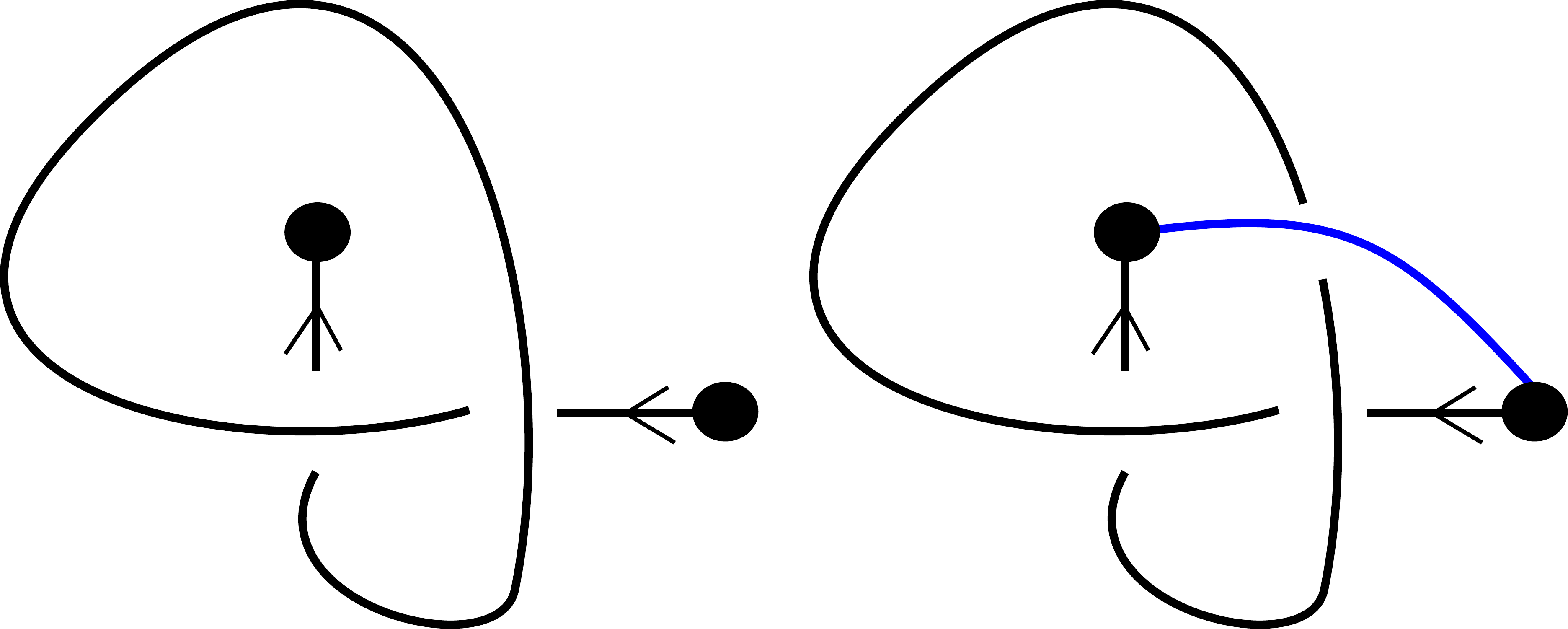}
      \caption{The knotoid $2_1$ and its overpass closure (the positive trefoil).}\label{fig:signedtrefoil}
\end{figure}
\section{More general models using 3-string tangles}\label{section:moregeneral3-stringmodels}
\subsection{Wagon wheel graph model}

In previous sections, we assumed that the protein originally takes the form of a rational closure of an unknotted 3-braid before the folding process. A 3-braid is a sub-collection of tangles that belong to the class of rational 3-string tangles. To obtain a new perspective, we can suppose that the initially, the protein takes the form of the unknot in a standard wagon wheel graph form.

\begin{figure}[ht!]
  \centering
    \includegraphics[width=0.5\textwidth]{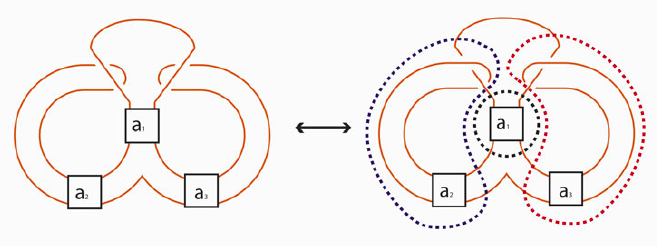}
      \caption{The wagon wheel graph model. The box $a_1$ is a vertical twist. The boxes $a_2$ and $a_3$ are horizontal twists.}\label{fig:wagon}
\end{figure}

\begin{res}\label{res:Montesinos}
    Let $T$ be a rational 3-string tangle in standard wagon wheel graph form. Then, the knot resulting from performing the $\Gamma_2$-move at the closure is a Montesinos knot $(\frac{2}{2a_2+1},\frac{1}{a_1},\frac{2}{a_3-1}).$
\end{res}
\begin{proof}
This follows from explicitly drawing the diagram. More specifically, after the move is performed, there are two additional half twists that can be combined with $a_2$ to yield a rational tangle $\frac{2}{2a_2+1}$. Also, there are two additional half twists that can be combined with $a_3$ to yield a rational tangle $\frac{2}{a_3+1}$. For instructions on how to compute the fraction corresponding to a rational tangle, the readers can consult Section \ref{Section:Taylor}
\end{proof}
We remark that Result \ref{res:Montesinos} gives Montesinos knots each differs from the unknot by one $\overrightarrow{\Gamma_2}$ move. Using this model, we get solutions that we cannot obtain from Theorem \ref{thm:main}. Namely, Montesinos knots are not necessarily 4-plat knots.
\subsection{3-braid closure model}
To obtain the granny knot, we consider a 3-braid model with a different closure. We will make use of the following classification result by Birman and Menasco \cite{birman1993studying}: If a closure of a 3-braid is the unknot, then the 3-braid is a conjugate of $\sigma_1\sigma_2, \sigma_1^{-1}\sigma_2^{-1},$ or $\sigma_1\sigma_2^{-1}.$ Note that Birman and Menasco use the notation $\sigma$, but we used the notation $E$ to mean the same thing when we introduced 3-braids earlier. 
\begin{res}\label{res:BirmanMenasco}
    Let $B$ be a 3-braid whose braid closure is the unknot. Then, the knot resulting from performing the $\Gamma_2$-move at the closure is a 3-braid knot where the braid can take one of the following three forms: $\mathcal{T}(-2,w,1,1,-w,2), \mathcal{T}(-2,w,-1,-1,-w,2),$ or $\mathcal{T}(-2,w,1,-1,-w,2).$
\end{res}
\begin{figure}[ht!]
  \centering
    \includegraphics[width=0.5\textwidth]{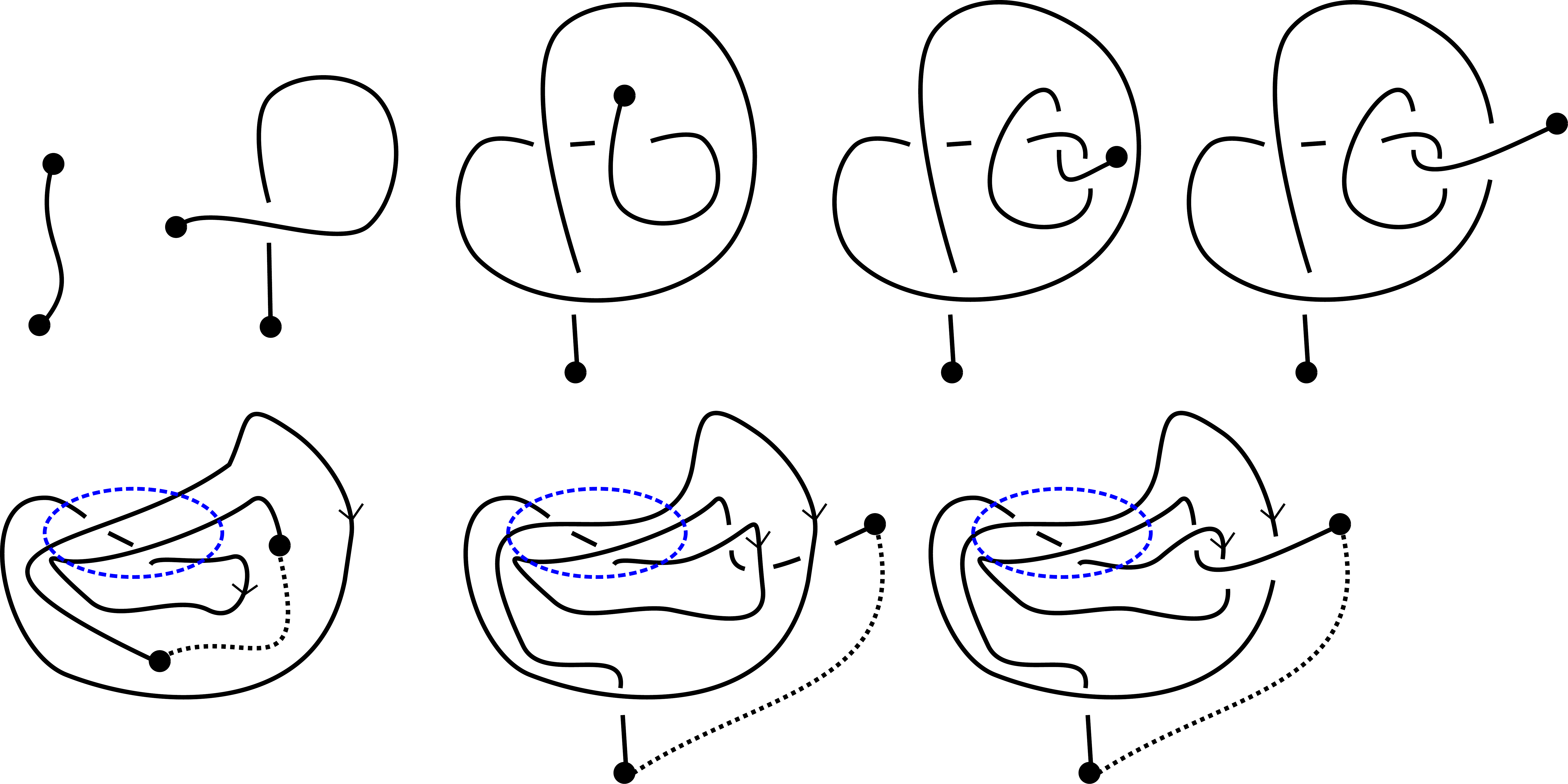}
      \caption{Getting the granny knot}\label{fig:granny}
\end{figure}

\begin{proof}
By a result of Birman and Menasco, the 3-braid $B$ whose closure is the unknot is conjugate to $\mathcal{T}(1,1),\mathcal{T}(-1,-1),$ or $\mathcal{T}(1,-1)$. In other words, the braid $B$ has the form $\mathcal{T}(w,1,1,-w), \mathcal{T}(w,-1,-1,-w),$ or $\mathcal{T}(w,1,-1,-w),$ where $w$ is an arbitrary braid word. The $-2$ and $2$ terms in the theorem statement comes from performing the $\Gamma_2$-move.
\end{proof}

For instance, when the 3-braid is $\mathcal{T}(1,-1),$ we get the trefoil. Result \ref{res:BirmanMenasco} gives us a way to construct 3-braid knots $K$ such that $d_{\overrightarrow{\Gamma_2}}(K,unknot) = 1$. Since the Stevedore's knot from B{\"o}linger et al.'s simulation is a 4-braid knot, it does not arise from the model in Result \ref{res:BirmanMenasco}, but we can get $3_1\# 3_1$.

\bibliographystyle{plain}
\bibliography{ref}

\end{document}